\documentclass[a4paper]{article}
\usepackage{fullpage}

\usepackage[roman]{complexity}
\usepackage{graphicx}
\usepackage{tikz}
\usepackage[ruled,vlined, linesnumbered]{algorithm2e}
\usepackage{xifthen}
\usepackage{amsmath,amssymb,amsthm,amsfonts}
\usepackage{hyperref}

\usepackage{import}
\usepackage{cleveref}

\newtheorem{theorem}{Theorem}
\newtheorem{lemma}{Lemma}
\newtheorem{definition}{Definition}
\newtheorem{corollary}{Corollary}

\newcommand{\Inc}[1]{\mathrm{IncP}_{#1}}
\newcommand{\UInc}[1]{\mathrm{UsualInc}_{#1}}
\newcommand{\WInc}[1]{\mathrm{WCInc}_{#1}}
\newcommand{\AmInc}[1]{\mathrm{AmInc}_{#1}}

\newcommand{\N}{\mathbb{N}}

\newcommand{\enum}[1]{\Pi_{#1}}

\newcommand{\EnumP}{\mathrm{ EnumP}}

\newcommand{\OutputP}{\mathrm{ OutputP}}
\newcommand{\IncP}{\mathrm{IncP}}

\newcommand{\DelayP}{\mathrm{DelayP}}

\newcommand{\AmDelayP}{\mathrm{AmDelayP}}
\newcommand{\SDelayP}{\mathrm{SDelayP}}

\newcommand{\timev}[3][]{\ifthenelse{\isempty{#2}}{t[#1]}{\ifthenelse{\isempty{#1}}{t_{#2}^{#3}}{t_{#2}^{#3}[#1]}}}

\renewcommand{\th}{\text{th}}
\newcommand{\Output}{\mathtt{OUTPUT}}
\newcommand{\test}[1]{\textsc{Check}\smash{\cdot}#1}
\newcommand{\ext}[1]{\textsc{ExtSol}\smash{\cdot}#1}

\newcommand{\ramload}{\mathsf{load}}
\newcommand{\ramstep}{\mathsf{move}}
\newcommand{\move}{\mathsf{move}}
\newcommand{\steps}{\mathsf{steps}}
\newcommand{\ramcopy}{\mathsf{copy}}

\NewDocumentCommand{\bbload}{O{\abbox}}{\ensuremath{\mathsf{load}(#1)}}
\NewDocumentCommand{\bbstep}{O{M}}{\mathsf{move}(#1)}
\NewDocumentCommand{\bbsteps}{O{M}}{\mathsf{steps}(#1)}
\NewDocumentCommand{\bbisoutput}{O{M}}{\mathsf{output_?}(#1)}
\NewDocumentCommand{\bboutput}{O{M}}{\mathsf{solution}(#1)}
\NewDocumentCommand{\bbcopy}{O{M}}{\mathsf{copy}(#1)}

\NewDocumentCommand{\bbam}{O{\abbox}}{{\overline{\mathsf{amdelay}}}({#1})}
\NewDocumentCommand{\bbsol}{O{\abbox}}{\overline{\sharp #1}}

\NewDocumentCommand{\bbsize}{O{\abbox}}{\overline{\mathsf{bitsize}}({#1})}
\NewDocumentCommand{\bsize}{O{\abbox}}{{\mathsf{bitsize}}({#1})}

\NewDocumentCommand{\adv}{O{\sigma} O{D} O{S}}{\mathcal{A}_{#1}(#2, #3)}

\NewDocumentCommand{\ladv}{O{\sigma}}{\mathcal{A}_{#1}}

\NewDocumentCommand{\bbspace}{O{\abbox}}{{\mathsf{space}({#1})}}

\newcommand{\inc}{\mathsf{inc}}

\newcommand{\mbit}{\mathsf{mbit}}

\newcommand{\queue}{\mathsf{queue}}
\newcommand{\queueinsert}{\mathsf{insert}}

\newcommand{\pull}{\mathsf{pull}}

\newcommand{\push}{\mathsf{push}}
\newcommand{\pop}{\mathsf{pop}}

\newcommand{\out}{\mathsf{output}}
\newcommand{\length}{\mathsf{length}}

\newcommand{\Counter}{\mathsf{Counter}}

\NewDocumentCommand{\abbox}{}{\ensuremath{\mathcal{I}}}
\NewDocumentCommand{\calM}{}{\ensuremath{\mathcal{M}}}

\NewDocumentCommand{\delay}{O{M} O{}}{
  \ifthenelse{\isempty{#2}}
  {\ensuremath{\mathsf{delay}(#1)}}
  {\ensuremath{\mathsf{delay}(#1,#2)}}
}
\NewDocumentCommand{\amdelay}{O{M} O{}}{\ensuremath{\mathsf{am}\delay[#1][#2]}}
\NewDocumentCommand{\avdelay}{O{M} O{}}{\ensuremath{\mathsf{avg}\delay[#1][#2]}}

\title{From Amortized to Worst Case Delay in Enumeration Algorithms} %TODO Please add

\usepackage{authblk}

\author[1]{Florent Capelli}

\affil[1]{Univ. Artois, CNRS, UMR 8188, Centre de Recherche en Informatique de Lens (CRIL), F-62300 Lens, France}%{capelli@cril.fr}{https://orcid.org/0000-0002-2842-8223}{}

\author[2]{Yann Strozecki}
\affil[2]{Université Paris Saclay, UVSQ, DAVID, France}
\date{}
\newtheorem{openproblem}{\bf Open problem}

\begin{document}
\maketitle

\begin{abstract}
The quality of enumeration algorithms is often measured by their delay, that is, the maximal time spent between the output of two distinct solutions. If the goal is to enumerate $t$ distinct solutions for any given $t$, then another relevant measure is the maximal time needed to output $t$ solutions divided by $t$, a notion we call the \emph{amortized delay} of the algorithm, since it can be seen as the amortized complexity of the problem of enumerating $t$ elements in the set. In this paper, we study the relation between these two notions of delay, showing different schemes allowing one to transform an algorithm with polynomial amortized delay for which one has a blackbox access into an algorithm with polynomial delay. We complement our results by providing several lower bounds and impossibility theorems in the blackbox model.

% In this paper, we introduce a technique we call geometric amortization for enumeration algorithms, which can be used to make the delay of enumeration algorithms more regular with little overhead on the space it uses. More precisely, we consider enumeration algorithms having incremental linear delay, that is, algorithms enumerating, on input $x$, a set $A(x)$ such that for every $t \leq \sharp A(x)$, it outputs at least $t$ solutions in time $O(t \cdot p(|x|))$, where $p$ is a polynomial. We call $p$ the incremental delay of the algorithm. While it is folklore that one can transform such an algorithm into an algorithm with maximal delay $O(p(|x|))$, the naive transformation may use exponential space.  We show that, using geometric amortization, such an algorithm can be transformed into an algorithm with delay $O(p(|x|)\log(\sharp A(x)))$ and space $O(s\log(\sharp A(x)))$ where $s$ is the space used by the original algorithm. In terms of complexity, we prove that classes $\DelayP$ and $\Inc{1}$ with polynomial space coincide. 

% We apply geometric amortization to show that one can trade the delay of flashlight search algorithms for their average delay up to a factor of $O(\log(\sharp A(x)))$. We illustrate how this tradeoff is advantageous for the enumeration of solutions of DNF formulas.%%% Local Variables:
%%% mode: latex
%%% TeX-master: "theoretics"
%%% End:

\end{abstract}

\newpage

\section{Introduction}

An enumeration problem is the task of listing a set of elements without redundancies. It is an important and old class of problems: the Baguenaudier game~\cite{lucas1882recreations} from the $19$th century can be seen as the problem of enumerating integers in Gray code order. Ruskey even reports~\cite{ruskey2003combinatorial} on thousand-year-old methods to list simple combinatorial structures such as the subsets or the partitions of a finite set. Modern enumeration algorithms date back to the $1970$s with algorithms computing circuits or spanning trees of a graph~\cite{tiernan1970efficient,read1975bounds}, while fundamental complexity notions for enumeration have been formalized $30$ years ago by Johnson, Yannakakis and Papadimitriou~\cite{johnson1988generating}.  
The main specificity of enumeration problems is that the size of the enumerated set is typically exponential in the size of the input.
Hence, a problem is considered tractable and said to be \emph{output polynomial} when it can be solved in time polynomial in the size of the \emph{input and the output}.
 This measure is relevant when one wants to generate and store all elements of a set, for instance to build a library of objects later to be analyzed by experts, as it is done in biology, chemistry, or network analytics~\cite{andrade2016enumeration,barth2015efficient,bohmova2018computing}.

 For most problems, the set to enumerate is too large, or may not be needed in its entirety. It is then desirable to efficiently generate a part of the set for statistical analysis or on the fly processing. In this case, a more relevant measure of the complexity and hence of the quality of the enumeration algorithm is its \emph{delay}, that is, the time spent between two consecutive outputs. One prominent focus has been to design algorithms whose delay is bounded by a polynomial in the size of the input. Problems admitting such algorithms constitute the class $\DelayP$ and many problems are in this class, for example the enumeration of the maximal independent sets of a graph~\cite{johnson1988generating}, or answer tuples of restricted database queries~\cite{DBLP:journals/tocl/DurandG07} (see~\cite{EnumOfEnum} for many more examples). 

 It also happens that new elements of the output set, also called solutions, become increasingly difficult to find. In this case, polynomial delay is usually out of reach but one may still design algorithms with \emph{polynomial incremental time}. An algorithm is in  polynomial incremental time if for every $i$, the delay between the output of the $i^\th$ and the $(i+1)^{\mathsf{st}}$ solution is polynomial in $i$ and in the size of the input. Such algorithms naturally exist for saturation problems: given a set of elements and a polynomial time function acting on tuples of elements, produce the closure of the set by the function. One can generate such a closure by iteratively applying the function until no new element is found. As the set grows bigger, finding new elements becomes harder. %For instance, one can generate all elements of a finite group by iteratively applying the group operation on a finite set of generators. 
 The best algorithm to generate circuits of a matroid uses a closure property of the circuits~\cite{khachiyan2005complexity} and is thus in polynomial incremental time. The fundamental problem of generating the minimal transversals of a hypergraph can also be solved in quasi-polynomial incremental time~\cite{fredman1996complexity, blasius2019efficiently} and some of its restrictions in polynomial incremental time~\cite{eiter2003new}. 
% The main drawback of these saturation algorithms is the need to store all output solutions, which often makes them impractical.% ref?
%In this paper, the class of problems which can be solved with polynomial incremental time is denoted by $\WCIncP$.

 While the delay and incremental delay are natural ways of measuring the quality of an enumeration algorithm, it might sometimes be too strong of a restriction. Indeed, if the enumeration algorithm is used to generate a subset of the solutions, it is often enough to have guarantees that the time needed to generate $i$ solutions is reasonable for every $i$. For example, one could be satisfied with an algorithm that has the property that after a time $i\cdot p(n)$, it has output at least $i$ solutions, where $p$ is a polynomial and $n$ the input size. In this case, drawing parallel with the complexity analysis of data structures, we will say that $p(n)$ is the \emph{amortized delay} of the algorithm\footnote{The conference version~\cite{CapelliS23} of this paper refered to \emph{linear incremental delay}. We decided to change to a more explicit name despite possible conflicts in the existing literature, see \Cref{sec:preliminaries} for more details.}.

While polynomial delay enumerator also have polynomial amortized delay, the converse is not true, as the delay implicitly enforces some kind of regularity in the way solutions are output. Take for example an algorithm that, on an input of size $n$, outputs $2^n$ solutions in $2^n$ steps, then nothing for $2^n$ steps and finally outputs the last solution. It can be readily verified that this algorithm outputs at least $i$ solutions after $2i$ steps for every $i \leq 2^n+1$, that is, this algorithm has amortized delay $2$. However, the delay of such an algorithm is not polynomial as the time spent between the output of the last two solutions is $2^n$. Instead of executing the output instruction of this algorithm, one could store the solutions that are found in a queue. Then, every two steps of the original algorithm, one solution is removed from the queue and output. The amortized delay being $2$ ensures that the queue is never empty when dequeued and we now have polynomial delay. Intuitively, the solutions being dense in the beginning, they are used to pave the gap until the last solution is found. While this strategy may always be applied to turn an enumerator with polynomial amortized delay into a polynomial delay algorithm, the size of the queue may become exponential in the size of the input. In the above example, after the simulation of $2^n$ steps, $2^n$ solutions have been pushed into the queue but only $2^{n-1}$ of them have been dequeued. Hence the queue contains $2^{n-1}$ solutions. 
Unfortunately, algorithms that require exponential space are impractical for real-world applications. Consequently, significant effort has been invested in developing polynomial-delay methods with polynomial space~\cite{lawler1980generating,avis1996reverse,cohen2008generating,conte2019new,DBLP:conf/icalp/BrosseLM22} and even to obtain a space complexity sublinear in the input size, such as in the generation of maximal cliques within massive networks~\cite{conte2016sublinear}.

\subparagraph{Contributions.} In this paper, we conduct an in-depth study of the following problem: can any algorithm with amortized delay $p(n)$ be transformed into an algorithm with delay $d(n)$ that is polynomial in $p(n)$? We call such a transformation a \emph{regularization scheme}, because we go from an algorithm that may have large gaps between outputing two distinct solutions to an algorithm with a guaranteed delay. From the discussion above, it is clear that regularization is always possible but it comes at an exponential price in the memory footprint. Moreover, the naive regularization scheme sketched above works only if the amortized delay $p(n)$ of the original algorithm is known, or if at least an upper bound on it is available. In this paper, we explore  more efficient regularization schemes and their limitations. Our first contribution is a regularization scheme similar to the naive one but where the knowledge of $p(n)$ is not necessary. We also observe that the naive method can be made to work with polynomial space in some particular cases and apply it to the problem of enumerating the models of a DNF formula. 

Our second and main contribution is the design of a regularization scheme that uses only a polynomial amount of memory. In other words, we show that the class $\DelayP^\poly$---problems solvable by a polynomial delay and polynomial space algorithm---is the same as the class $\AmDelayP^\poly$---problems solvable by a polynomial space polynomial amortized delay algorithm. In other words, we prove $\DelayP^\poly = \AmDelayP^\poly$, answering positively a question we raised in~\cite{capelli2019incremental} and where only special cases were proven. This finding indicates that, from a purely theoretical standpoint, polynomial amortized delay and polynomial delay are actually the same notions, in the sense that they define the same class of ``easy'' problems which came as a surprise to us. 
Our result relies on a constructive method that we call \emph{geometric regularization}. Our first description of geometric regularization only works when it is given an upper bound on the amortized delay and on the number of solutions of the original algorithm. We later lift these requirements by presenting improved versions that can adapt to unknown solutions number or amortized delay.

The main drawback of geometric regularization is that it modifies the output order. While we propose a way of circumventing this limitation, we complement our algorithmic results by lower bounds, showing that this is very likely unavoidable. In particular, we show that any regularization scheme that preserves the order of solutions and only have a blackbox access to the original algorithm must either have exponential delay or use exponential space. We also show that regularization schemes cannot have a delay that is linear in the amortized delay of the original algorithm (unless the amortized delay is given explicitly as input), showing that the delay obtained by our adaptative schemes is nearly optimal in the blackbox model.

\subparagraph{Related work.} The notion of polynomial amortized delay is natural enough to have appeared before in the literature. In her PhD thesis~\cite{Goldberg91}, Goldberg introduced the notion of \emph{polynomial cumulative delay}, which exactly corresponds to our notion of polynomial amortized delay. We however decided to stick to the terminology of~\cite{capelli2019incremental}. Goldberg mentions on page 10 that one can turn a polynomial amortized algorithm into a polynomial delay algorithm but at the price of exponential space. She argues that one would probably prefer in practice polynomial amortized delay and polynomial space over polynomial delay and exponential space. Interestingly, she also designs for every constant $k$, an algorithm with polynomial amortized delay using polynomial space to enumerate on input $n$, every graph that is $k$-colorable (Theorem 15 on page 112). She leaves open the question of designing a polynomial delay and polynomial space algorithm for this problem, which now comes as a corollary of our theorem.

In~\cite{TGR22}, Tziavelis, Gatterbauer and Riedewald introduce the notion of \emph{Time-To-$k$} to describe the time needed to output the $k$ best answers of a database query for every $k$. Graph problems have also been studied with the same complexity measure, under the name $k$-best enumeration~\cite{eppstein2015k}.
They design algorithms having a Time-To-$k$ complexity of the form $\poly(n)k$ where $n$ is the size of the input, which hence corresponds to our notion of amortized delay. They argue that delay is sufficient but not necessary to get good Time-to-$k$ complexity and they argue that in practice, having small Time-to-$k$ complexity is better than having small delay. Observe however that in their case, our method does not apply well since they are interested in the \emph{best} answers, meaning that the order is important in this context. Our method does not preserve order. 

%Moreover, it is well-known that by using exponential memory, one can eliminate duplicates in the output (see for example the Cheater's Lemma in~\cite{carmeli2019enumeration}), our technique does not allow this: if the original enumeration algorithm outputs $d$ times the same solution, then our transformation will also output it $d$ times.

\subparagraph{Organization of the paper.}

In \Cref{sec:preliminaries}, we introduce enumeration problems, the related computational model, enumeration processes, and complexity measures. \Cref{sec:bufferamortization} presents buffering techniques to regularize the delay, which requires an exponential size memory. We present the folklore algorithm, and a new analysis of its complexity when taking the size of the solutions output into account. We also present a new variation of this method that works without prior knowledge of the amortized delay. We show that on a large class of problems, known as self-reducible problems, one can make these techniques more efficient with only a polynomial space overhead. \Cref{sec:mainalgo} shows how to transform the amortized delay of an enumeration process $\abbox$ into guaranteed worst case delay with additional memory logarithmic in the number of solutions, using a technique we call geometric regularization. Interactive visualization of how geometric regularization works can be found at \url{http://florent.capelli.me/coussinet/}. In \Cref{sec:geometric_amortization_unknown_delay}, we propose a variant of geometric regularization wich works without prior knowledge of the amortized delay, inspired by the method we previously designed for the buffering technique. 
In \Cref{sec:Inci}, we generalize geometric regularization to amortized incremental delay algorithms.  \Cref{sec:lower-bounds} proves lower bounds establishing the optimality of some of our results in the blackbox oracle model. In particular, we show that a regularization scheme without knowledge of the amortized delay $d$ of an enumeration process cannot have the optimal worst-case delay $O(d)$. Moreover, we show that an order preserving regularization scheme must have the product of its space and delay exponential in the size of the enumerated solutions. To outline the main ideas of our algorithms, they are presented using pseudocode with instructions to run an enumeration process. This enumeration process is given by a Random Access Machin (RAM) that must be simulated. The details on the complexity of implementing these instructions with constant time and space overhead are given in \Cref{sec:oracles}. 

%%% Local Variables:
%%% mode: latex
%%% TeX-master: "theoretics"
%%% End:

\section{Preliminaries}
\label{sec:preliminaries}

\subparagraph{Enumeration problems.}

Let $\Sigma$ be a finite alphabet and $\Sigma^*$ be the set of finite words built on $\Sigma$.
We denote by $|x|$ the length of $x \in \Sigma^*$.
Let $A\subseteq \Sigma^{*}\times\Sigma^{*}$ be a binary predicat. We write $A(x)$ for the set of $y$ such that $A(x,y)$ holds. The enumeration problem $\enum{A}$ is the function which associates $A(x)$ to $x$. The element $x$ is called the \emph{instance} or the \emph{input}, while an element of $A(x)$ is called a \emph{solution}. We denote the cardinality of a set $A(x)$ by $\sharp A(x)$.

A predicate $A$ is said to be \emph{polynomially balanced} if for all $y \in A(x)$, $|y|$ is polynomial in $|x|$. It implies that $\sharp A(x)$ is bounded by $|\Sigma|^{\poly(|x|)}$. Let $\test{A}$ be the problem of deciding, given $x$ and $y$, whether $y \in A(x)$. The class $\EnumP$, a natural analogous to $\NP$ for enumeration, is defined to be the set of all problems $\enum{A}$  where $A$ is polynomially balanced and $\test{A} \in \P$. More details can be found in~\cite{capelli2019incremental,strozecki2019enumeration}.
% Most studied enumeration problems consist of listing solutions of $\NP$ problems, hence, in many cases, $A$ is polynomially balanced and $\test{A} \in \P$.  %Observe that if $A$ is in $\EnumP$, the problem of deciding given $x$ whether $A(x)$ is not empty is in $\NP$. The problem of computing $\sharp A(x)$ is in $\sharp P$.

\subparagraph{Computational model.}  

In this paper, we use the Random Access Machine (RAM) model introduced by Cook and Reckhow~\cite{cook1973time} with comparison, addition, subtraction and multiplication as its basic arithmetic operations augmented with an $\Output(i,j)$ operation which outputs the content of the values of registers $R_i, R_{i+1}, \dots, R_j$ as in~\cite{Bagan09,phdstrozecki} to capture enumeration problems. Our model is a hybrid between \emph{uniform cost model} and \emph{logarithmic cost model} (see~\cite{cook1973time,aho1974design}): addition, multiplication and comparison are in constant time on values less than $n$, where $n$ is the size of the input. In first-order query problems, it is justified by bounding the values in the registers by $n$ times a constant~\cite{DBLP:journals/tocl/DurandG07,Bagan09}. However, for general enumeration algorithms which may store and access $2^{n}$ solutions, dealing with large integers is necessary. Instead of bounding the register size, we define the cost of an instruction to be the sum of the size of its arguments \emph{divided by $\log(n)$}. Thus, any operation on a value polynomial in $n$ can be done in constant time, but unlike in the usual uniform cost model, we account for the cost of handling superpolynomial values.

We assume the $\Output(i,j)$ operation has constant cost, allowing algorithms where the delay between outputs is less than the size of the output. For example, the next solution can be obtained by only changing a few bits of the current one, as in Gray codes for enumerating integers. In this case, the time to actually write the output is greater than the time to generate it. To reflect this complexity, we consider the output operation to be free. However, as we shall see, when we use other oracles to interact with the output solutions, then we will consider that reading this solution depends on its size.

A RAM $M$ on input $x \in \Sigma^{*}$ produces a sequence of outputs $y_{1}, \dots, y_{s}$. The set of outputs of $M$ on input $x$ is denoted by $M(x)$ and its cardinality by $\sharp M(x)$. We say that $M$ solves $\enum{A}$ if, on every input $x \in \Sigma^{*}$, $ A(x) = M(x) $ and for all $i\neq j$ we have $y_{i} \neq y_{j}$, that is \emph{no solution is repeated}. All registers are initialized with zero. The space used by $M$ is the sum of the bitsizes of the integers stored in its registers, up to the register of the largest index accessed.

\subparagraph{Enumeration process and delay.} To measure the quality of enumeration of algorithm, we can consider the \emph{total time} it takes to enumerate all solutions.  Since the number of solutions can be exponential in the \emph{input} size, it is more relevant to give the total time as a function of the combined size of the \emph{input and output}.  However, this notion does not capture the dynamic nature of an enumeration algorithm. When generating all solutions already takes too long, we want to be able to generate at least some solutions. Hence, we should measure (and bound) the total time used to produce a given number of solutions. To this end, we introduce a few notions of delay. It is useful to introduce these notions for a process that does not depend on an input $x$, so we start by abstracting the computation of an enumeration algorithm into the following setting.

An \emph{enumeration process $\abbox$} for a set $S$ is a mapping $\N \rightarrow S \cup \{\bot, \square\}$ where for every $s \in S$, there is a unique $t \in \N$ such that $\abbox(t) = s$. Moreover, there exists $t_0$ such that $\abbox(t) = \square$ if and only if  $t \geq t_0$. A RAM machine $M$ on input $x \in \Sigma^*$ naturally induces an enumeration process $\abbox$ for $M(x)$, where $\abbox(t) = \bot$ if no solution is output at time $t$ by $M$ on input $x$, $\abbox(t)=s$ if solution $s$ is output at time $t$ and $\abbox(t) = \square$ if $M$ has stopped before time $t$ on input $x$. If $\abbox(t) = s \in S$, we say that $\abbox$ outputs $s$ at time $t$. 

Given an enumeration process $\abbox$ for a set $S$ of size $m$, we denote by $T_\abbox(i)$ the time where $\abbox$ outputs the $i^\th$ element of $S$. By convention, we let $T_\abbox(0)=0$. When the process is clear from the context, we drop the subscript $\abbox$ and write $T(i)$. The \emph{delay} of $\abbox$, denoted as $\delay[\abbox]$, is defined as $\max_{0 \leq i < |S|}{T(i+1)-T(i)}$, that is, the delay of $\abbox$  is the maximal time one has to wait between two outputs.

In this paper, we are also interested in other notions of delay. The central notion of delay we will be using is the notion of \emph{amortized delay}\footnote{In the conference version of this paper~\cite{CapelliS23}, this notion was called \emph{linear incremental delay}. We decided to change the name to amortized delay to give a better intuition on what it is and match the terminology of complexity for data structures (see below).}. The \emph{amortized delay} of $\abbox$, denoted as $\amdelay[\abbox]$, is defined as $\max_{1 \leq i \leq |S|} {T(i)\over i}$. Amortized delay is an interesting notion because it can be used to bound the time needed to output $i$ solutions for every $i$. Indeed, if we know that $\amdelay[\abbox] \leq d$, then we know that $\abbox$ outputs at least $i$ solutions after $i \cdot d$ steps. 

  The \emph{average delay} of $\abbox$, denoted as $\avdelay[\abbox]$, is defined as $T(|S|) \over |S|$. Average delay does not say anything on the dynamics of the computation and an algorithm can have a small average delay while the time needed to find a solution is extremely high, for example, in the case where all solutions are output at the very end of the computation.
  Previous work sometimes uses the term amortized delay to define what we call average delay in this paper~\cite{ferreira2014amortized}. In this paper, we decided to differentiate both notions at the risk of having a small confusion with earlier work. We believe this terminology to be the best one as it nicely connects with classical notions in data structures. Indeed, seeing an enumeration algorithm as some sort of process that extract solutions from a data structure, the delay corresponds to the worst case complexity of extracting one element from the data structure, while the amortized delay corresponds to the amortized complexity of extracting $i$ solutions from the data structure. In this paper, we sometimes use the term ``worst-case delay'' in place of ``delay'' to be even closer to this terminology.

  The three notions are related by the following inequality:
  \begin{theorem}
    \label{thm:cmpdelay}
    For every enumeration process $\abbox$, we have $\avdelay[\abbox] \leq \amdelay[\abbox] \leq \delay[\abbox]$.
  \end{theorem}
  \begin{proof}
    The first inequality is trivial since the value $T(|S|)\over |S|$ is clearly smaller than $\max_{1 \leq i \leq |S|} {T(i)\over i}$. The second inequality comes from the following calculation for every $i$:

    \begin{align*}
      T(i) & = \sum_{j=0}^{i-1} T(j+1)-T(j) \\
                            & \leq i \cdot \delay[\abbox]           
    \end{align*}
    Hence $\amdelay[\abbox] = \max_{i} {T(i) \over i}  \leq \delay[\abbox]$.
  \end{proof}

  These notions of delay are sometimes defined with an offset of $T(1)$, that is, the time before the output of the first solution may be greater than the rest, allowing preprocessing. We decided to ignore this asymmetry in this paper to focus on the core techniques. All our algorithms can be adapted to this setting, see \Cref{sec:conclusion} for more details. Similarly, \emph{postprocessing} -- the time spent between the output of the last solution and the end of the computation -- is sometimes dealt with differently. To simplify the presentation, we assume that there is no postprocessing, that is, a RAM solving an enumeration problem stops right after having output the last solution. This assumption does not affect the complexity classes studied in this paper, as the output of the last solution can be delayed to the end of the algorithm or we could decide to output an additional solution right before stopping.

We naturally generalize this notion to RAM machine. We denote by $\delay[M][x]$, $\amdelay[M][x]$ and $\avdelay[M][x]$ the worst case (respectively amortized, average) delay of the enumeration process induced by $M$ on input $x$. The delay of machine $M$, denoted by $\delay$, is defined as the function mapping any input size $n$ to the value  $\max_{x, |x|=n} \delay [M][x]$. We similarly define the amortized delay $\amdelay$, average delay $\avdelay$ and preprocessing of $M$. 

\subparagraph{Complexity classes.} These notions of delay induce natural complexity classes. Some of them have been originally introduced by Johnson, Yanakakis and Papadimitriou in~\cite{johnson1988generating}. We are mostly interested in two classes.

\begin{definition}[Polynomial Amortized Delay]
  The class $\AmDelayP$ is the class of enumeration problem $\enum{A} \in \EnumP$ such that there exists a polynomial $d$ and a machine $M$ which solves $\enum{A}$ such that for every $x$, $\amdelay[M][x] \leq d(|x|)$.% and the preprocessing $T_M(x,1)$ of $M$ is bounded by $d(|x|)$.
\end{definition}

\begin{definition}[Polynomial Delay]
  The class $\DelayP$ is the class of enumeration problem $\enum{A} \in \EnumP$ such that there exists a polynomial $d$ and a machine $M$ which solves $\enum{A}$ such that for every $x$, $\delay[M][x] \leq d(|x|)$. % and the preprocessing $T_M(x,1)$ of $M$ is bounded by $d(|x|)$.
\end{definition}

Observe that by \Cref{thm:cmpdelay}, we directly have $\DelayP \subseteq \AmDelayP$.  Polynomial delay is the most common notion of tractability in enumeration, because it guarantees both regularity and linear total time and also because it is relatively easy to prove that an algorithm has a polynomial delay. Indeed, most methods used to design enumeration algorithms such as backtrack search with a polynomial time extension problem~\cite{MaryS19}, or efficient traversal of a supergraph of solutions~\cite{lawler1980generating,avis1996reverse,conte2019listing}, yield polynomial delay algorithms on many enumeration problems.

To better capture the notion of tractability in enumeration, it is important to use polynomial space algorithms. We let $\DelayP^{\poly}$ be the class of problems solvable by a machine in polynomial space and polynomial delay. 
We define $\AmDelayP^{\poly}$, as the class of problems which can be solved by a machine in polynomial space and polynomial amortized delay\footnote{This class was denoted $\IncP^{\poly}$ in the conference version~\cite{CapelliS23} of this paper.}.

\subparagraph{Blackbox oracle access.}  It is easy to see that an algorithm with amortized delay $p(n)$ can have arbitrarily large worst-case delay. This paper is focused in understanding when and how one can transform an algorithm with amortized delay $p(n)$ into an algorithm with guaranteed worst-case delay $\poly(p(n))$, a task that we will refer to as \emph{regularization}. To describe such transformations, it will be handy to have a formalism where we can describe how to interact with an existing enumeration process $\abbox$ given as a blackbox. A \emph{RAM with blackbox access to $\abbox$} is a standard RAM with the following blackbox access to the enumeration process $\abbox$: 
 \begin{enumerate}
 \item  $\bbam[\abbox]$, $\bbsol[\abbox]$ and $\bbsize[\abbox]$ respectively return an upper bound on the amortized delay, the number of solutions and the maximal size of solutions  of $\abbox$.
\item $M \gets \bbload[\abbox]$ creates a new simulation $M$ of the enumeration process $\abbox$,
% \item $\bbcopy[p]$ returns an independent copy of $p$,
\item $\bbsteps[M]$ returns the number of steps simulated by $M$ so far, with 
$\bbsteps[M] = 0$ when $M$ is created,
\item $\bbstep[M]$ increments $\bbsteps[M]$ by $1$%executes one step of simulation $M$
, and returns False if $\abbox(\bbsteps[M]) = \square$,
%the computation of $M$ is finished,
\item $\bbisoutput[M]$ returns true if and only if $\abbox(\bbsteps[M]) \in S$, where $S$ is the set of solutions output by $\abbox$,
\item $\bboutput[M]$  which outputs $\abbox(\bbsteps[M])$ if $\bbisoutput[M]$ is true and fails otherwise.
\end{enumerate}

In this paper, we consider that each blackbox operation is performed with time and space complexity $O(1)$.
% where $s$ is a known upper bound on the space used by $\abbox$. 
For $\bbam, \bbsol$ and $\bbsize$, we assume these values are provided along the oracle and hence, can be accessed by reading a register. It is a reasonable assumption in many cases, since such upper bounds are often known for a given enumeration algorithm. 
% Concerning the complexity of the other operations, constant time may seem optimistic but can be achieved  when an upper bound $s$ on the space used by $\abbox$ is known. Indeed, if so, we can allocate $s$ registers for each simulation of $\abbox$. 
Every instruction above but $\bbsteps$ could be implemented right away in constant time. To implement $\bbsteps$, it is more delicate in our complexity model: if implemented as a counter, it may become too big for us to assume that incrementing the counter can be done in constant time. This can be circumvented by implementing the counter using Gray code, where incrementing can now be done in constant time. In some algorithm, we need more complex arithmetic operations on such counters (division or comparison). While we cannot implement these operations in constant time, Gray code still allows to get the most significant bit of the counter in constant time, which is enough for our purposes. Hence, to simplify analysis of our algorithms, we assume in this paper that maintaining a counter for $\bbsteps$ and using it in our algorithms can be performed with complexity $O(1)$. Details on how to implement Gray code counters and how to use them exactly so that our algorithms still work with the announced complexity are given in \Cref{sec:oracles}. The space necessary to implement all instructions is not $O(1)$ as we must at least store the configuration of the simulated RAM. We explain how to implement all instructions with a constant overhead in space in \Cref{sec:oracles}.

Given a subset of blackbox instructions $B$, we say that a machine is independent of $B$ if it does not use any instruction in $B$. For example, a machine is independant from $\bbam$ if it runs without calling instruction $\bbam$. 

% Indeed, a $\bbcopy$ instruction needs to copy somehow the current state of $p$ which may use exponential space and would be very costly. Similarly, $\bbsteps$ may return an integer that has size too big to be used in constant time in our RAM model. Moreover, $s$ is unknown, then one would need to dynamically allocate memory which would incur additional costs for the $\bbstep$ instruction. We give details on the complexity of each operation depending on what we know about $M$ in \Cref{sec:oracles}.  \flo{yann, check this paragraph, please}

\subparagraph{Regularization scheme.} In this paper, we study the following question: given an enumeration process $\abbox$, can we transform it into an enumeration process $\abbox'$ which outputs the same set and such that $\delay[\abbox'] = \poly(\amdelay[\abbox])$? We study several variations of this question and show limits on what we can achieve with it. To this end, it is handy to have a name for such transformation. A \emph{regularization scheme} $\calM$ is a RAM with blackbox access such that for every enumeration process $\abbox$ given to it as a blackbox, $\calM$ outputs the same set of solutions as $\abbox$. Moreover, $\calM$ is said to be \emph{order preserving} if for every enumeration process $\abbox$ given to $\calM$ as a blackbox, $\calM$ outputs the same set as $\abbox$, in the same order.

%We consider that $\calM$ can adapt its strategy depending on several parameters of $\abbox$. Unless specified otherwise, we consider that $\calM$ can access the following information about $\abbox$: an upper bound $\amdelay[\abbox]$  $(p,b,N,s)$ in our model: $p$ is a guaranteed upper bound on  the amortized delay $p$ of $\abbox$,  $b$ and $N$ are upper bounds the size and number of solutions output by $\abbox$ respectively. We say that $\calM$ has \emph{regularization delay} $f(p,b,N)$ for $f \colon \N^3 \to \N$ if the delay of $\calM$ is bounded by $f(p,b,N)$ when given a black box access to a process $\abbox$ which outputs $N$ solutions of size at most $b$ with amortized delay and preprocessing at most $p$. Similarly, we say that $\calM$ has \emph{space overhead} $h(p,b,N)$ for $h \colon \N^3 \to \N$ if $\calM$ uses simultaneously at most $h(p,b,N,s)$ distinct pointers to the blackbox $\abbox$ (obtained via operator $\bbload[\cdot]$ or $\bbcopy[\cdot]$) where $\abbox$ outputs $N$ solutions of size at most $b$ with amortized delay and preprocessing at most $p$. If the computation of $\calM$ does not depend on a parameter $x \in \{p,b,N\}$, we say that it is independant from $x$. %\flo{Actually, now we can be more precise on the parameters used by $\calM$, e.g. in \Cref{thm:lower_bound_udelay}, I think we can allow it to use $(b,N,s)$ anyway.}

%%% Local Variables:
%%% mode: latex
%%% TeX-master: "theoretics"
%%% End:

\section{Buffer Based Regularization}
\label{sec:bufferamortization}
  
In this section, we present a naive, folklore regularization scheme whose delay is $O(\bbam)$ for any enumeration process $\abbox$. To do so, we store each solution in a queue and output solutions pulled from the queue with a guaranteed worst case delay. In other words, we show that we can transform the amortized delay of an enumeration algorithm into a (worst case) delay. This scheme is enough to establish that $\AmDelayP=\DelayP$ but it has two shortcomings. The first one is that the approach only works when the given upper bound $\bbam$ of $\amdelay[\abbox]$ is known. We hence give a regularization scheme that do not use $\bbam$ and still achieves a worst case delay $O(d_\abbox (\log d_\abbox)^2)$ where $d_\abbox$ is the real amortized delay of $\abbox$. That is, this scheme almost optimally turns an amortized delay into a worst case delay without any knowledge on the amortized delay of $\abbox$.

The second shortcoming is that the space used by both schemes is large, up to $O(\#\abbox)$, which in many cases, may be exponential in the space originally used by $\abbox$. In \Cref{sec:amort-self-reduc}, we revisit an unpublished technical report by Uno~\cite{uno03} and show that the buffer based technique may be implemented with polynomial space in a restricted case. We will improve on this aspect in \Cref{sec:mainalgo} by introducing a new regularization technique, that is not based on using a buffer. 

\subsection{Classical Method}
\label{sec:classical}

\SetKw{Break}{break}
\SetKwInOut{KwInput}{Input}
\SetKwInOut{KwOutput}{Output}

\begin{algorithm}
  \KwInput{Enumeration process $\abbox$.}
  \KwOutput{Enumerate $\abbox$ with delay bounded by $O(\bbam)$}
 \Begin{
   $Q \gets \queue(\emptyset)$, $B \gets 0$, $M \gets \bbload$\;
   \While{$\bbstep$} {
     $B \gets B+1$\;
     \lIf{$\bbisoutput$}{
        Insert $\bboutput$ in $Q$ \label{line:insert}
        }
     \If{$B = \bbam$}{
       $\out(\pull(Q))$\; \label{line:pull}
       $B \gets 0$\;
       }
     }
     \lWhile{$Q \neq \emptyset$}{$\out(\pull(Q))$}
  }
  \caption{Folklore regularization scheme with worst case delay $O(\bbam)$}
  \label{alg:enumqueue}
\end{algorithm}

The folklore regularization scheme (e.g.~\cite{johnson1988generating,phdstrozecki,carmeli2019enumeration}) is presented in Algorithm~\ref{alg:enumqueue}. The correctness of \Cref{alg:enumqueue} is easy to see: every solution will be eventually stored in the queue at some point and pulled from it, maybe in the last step. However, one has to prove that each time $\pull(Q)$ is called, the queue is not empty. This can be seen as follows: when $\pull(Q)$ is called for the $k^\th$ time, at least $k \cdot \bbam \geq k \cdot \amdelay[\abbox]$ steps of $\abbox$ have been simulated. Hence $\abbox$ have output at least $k$ solutions, that have all been pushed in $Q$ at some point. 

We now provide a precise and, to the best of our knowledge, new analysis of the complexity of the scheme when taking the size of the solutions output by $\abbox$ into account. Indeed, the delay of Algorithm~\ref{alg:enumqueue} is $O(\bbam \cdot b)$ where $b$ is a bound on the sizes of output solutions since each solution output by $\abbox$ have to be stored in the queue and we cannot assume this operation to be independent on $b$. In most applications, $b$ is less than the minimum time between two outputs. In this case, the cost of insertion in the queue can be amortized in the main loop of Algorithm~\ref{alg:enumqueue}, for example, by interleaving queuing process with loop iteration: for each iteration of the loop, we execute a constant amount of operations needed to push a solution in the queue. Since no new solution can be output by $\abbox$ in less than $b$ steps, there will never be more than one queuing process interleft with the loop. Hence, in this case, \Cref{alg:enumqueue} has $O(\bbam)$ delay. We have proven:

\begin{lemma}[Folklore]\label{lemma:naivescheme-bad}
Let $\abbox$ be an enumeration process. Then \Cref{alg:enumqueue} is a regularization scheme with worst case delay $O(\bsize \cdot \bbam)$. Moreover, if for every $i$, the time between the output of solutions $s_i$ and $s_{i+1}$ is at least $|s_i|$, then the worst case delay of \Cref{alg:enumqueue} is $O(\bbam)$.
\end{lemma}

\Cref{lemma:naivescheme-bad} is sufficient to prove that the class of problems from $\EnumP$ which can be solved with polynomial delay is the same as the class of problems from $\EnumP$ which can be solved with amortized polynomial delay.

\begin{theorem}
  \label{thm:delayp_is_amdelayp} $\DelayP = \AmDelayP$.
\end{theorem}
\begin{proof}
  The left to right inclusion is trivial by \Cref{thm:cmpdelay}: every algorithm with polynomial delay also have polynomial amortized delay. For the other inclusion, let $A \in \AmDelayP$ and let $I$ be an algorithm which on input $x$, outputs $A(x)$ with amortized delay $p(x)$. Since $A \in \EnumP$, there is a polynomial $q$ such that for every $x$ and $y \in A(x)$, $|y| \leq q(|x|)$. We built an algorithm enumerating $A(x)$ on input $x$ by using \Cref{alg:enumqueue} with enumeration process $\abbox = I(x)$ where we set $\bbam = p(|x|)$ and $\bbsize = q(|x|)$. By \Cref{lemma:naivescheme-bad}, this algorithm has delay $O((q\times p)(|x|))$ for every $x$, and hence, $A \in \DelayP$.
\end{proof}

We conclude this section by providing a better analysis of the complexity of this buffer based regularization scheme. In cases where the time spent between two outputs is small, for example, when a new solution is changed by modifying only small parts of a previous one and, the best bound we get for \Cref{alg:enumqueue} is $O(\bbam \cdot \bsize)$, because each push into the queue may take up to $\bsize$. We can however improve this bound by slightly modifying the enumeration process $\abbox$ so that the time between output $s_i$ and $s_{i+1}$ is at least $|s_i|$: we simply copy solution $s_i$ somewhere in the memory before proceeding with the execution of $\abbox$, as illustrated in \Cref{alg:pacing}.

\begin{algorithm}
  \KwInput{Enumeration process $\abbox$.}
  \KwOutput{Enumerate $\abbox$ such that delay between solutions $s_i$ and $s_{i+1}$ is at least $|s_i|$.}
  \Begin{
    $M \gets \bbload$\;
   \While{$\bbstep[M]$} {
     \If{$\bbisoutput[M]$}{
       $l \gets |\bboutput[M]|$\;
       Copy $\bboutput[M]$ into $R_i \dots R_{i+l}$\;
       $\Output(i, i+l)$\;
     }
      }
  }
  \caption{Spacing the solutions of $\abbox$.}
  \label{alg:pacing}
\end{algorithm}

\begin{lemma}\label{lemma:pacing}
  Let $\abbox$ be an enumeration process. Then \Cref{alg:enumqueue} enumerates the same set as $\abbox$. Moreover, the amortized delay of \Cref{alg:enumqueue} is $O(\amdelay[\abbox]+\bsize[\abbox])$ and for every $i$, the time between the output of $s_i$ and $s_{i+1}$ is at least $|s_i|$.
\end{lemma}
\begin{proof}
  Let $\abbox'$ be the new enumeration process given by \Cref{alg:pacing}. It is clear that $\abbox'$ enumerates the same set as $\abbox$ and that the time spent between output $s_i$ and $s_{i+1}$ is at least $|s_i|$ since it has to copy $s_i$ into $R_i,\dots,R_{i+l}$. Now, let $t_i$ be the time when $s_i$ is output by $\abbox$ and $t_i'$ the time when it is output by $\abbox'$. A straightforward induction on $i$ is enough to see that $t_i' \leq c \times (t_i + i\times \bsize)$ where $c$ is some constant overhead induced by the simulation. Hence, $\amdelay[\abbox'] \leq c \times (\max_i {t_i \over i} + \bsize) \leq O(\amdelay[\abbox]+\bsize)$.
\end{proof}

\begin{corollary}\label{prop:queue}
There exists a regularization scheme such that for every enumeration process $\abbox$, the worst case delay of the scheme is $O(\bbam+\bbsize)$.
\end{corollary}
\begin{proof}
  We simply run \Cref{alg:enumqueue} with enumeration process $\abbox'$, where $\abbox'$ is given by \Cref{alg:pacing} and where we set $\bbam[\abbox'] = \bbam+\bbsize$. By \Cref{lemma:pacing},  $\bbam[\abbox']$ is an upper bound on $\amdelay[\abbox']$. Moreover, the delay between output $s_i$ and $s_{i+1}$ by $\abbox'$ is at least $|s_i|$. Hence by \Cref{lemma:naivescheme-bad}, \Cref{alg:enumqueue} hence enumerates the solutions of $\abbox'$, that is, the solutions of $\abbox$,  with worst case delay $O(\bbam[\abbox']) = O(\bbam+\bbsize)$.
\end{proof}

% FLORENT : Add consequences for AmDelayP and DelayP.

\subsection{Adaptative Buffering}

We now improve the previous method by designing a regularization scheme that is independent from the upper bound $\bbam$ on the amortized delay of $\abbox$. \Cref{alg:enumqueue_improved} is a modified version of \Cref{alg:enumqueue} that automatically adapts the amortization to unknown amortized delay. The idea is to replace the exact amortized delay by a local approximation of it.

\begin{algorithm}[H]
    \KwInput{Enumeration process $\abbox$.}
  \KwOutput{Enumerate $\abbox$ with delay bounded by $O(p (\log p)^2)$ for $p = \amdelay[\abbox]$.}

  \Begin{
    $Q \gets \queue(\emptyset)$,
    $M \gets \bbload$,
    $S \gets 0$,
    $B \gets 0$,
    $d \gets 8$
    \;
    \While{$\bbstep$} {
      $B \gets B+1$\;
      \If{$\bbisoutput$}{
        $\queueinsert(\bboutput,Q)$\;
        $S \gets S+1$\;
      }
      $d = \max (d,\frac{\bbsteps}{S+1})$\;
      \If{$B \geq \lceil d\log(d)^2\rceil$}{\label{line:ceps}
        $\out(\pull(Q))$\; \label{line:pull2}
        $B \gets 0$\;
      }
    }
    \lWhile{$Q \neq \emptyset$}{
      $\out(\pull(Q))$
    }
   
  }
  \caption{Regularizing of an enumeration process $\abbox$ without knowledge on its amortized delay.}
  \label{alg:enumqueue_improved}
\end{algorithm}

\begin{theorem}\label{th:adaptative}
  Let $\abbox$ be an enumeration process with amortized delay $p = \amdelay[\abbox]$ and solution size $b = \bsize[\abbox]$. Then \Cref{alg:enumqueue_improved} is a regularization scheme with worst case delay $O(bp\log(p)^2)$. Moreover, if for every $i$, the time between the output of solutions $s_i$ and $s_{i+1}$ is at least $|s_i|$, then the worst case delay of \Cref{alg:enumqueue_improved} is $O(p\log(p)^2)$.
\end{theorem}

\begin{proof}
    The structure of Algorithm~\ref{alg:enumqueue_improved} is similar to the one of Algorithm~\ref{alg:enumqueue} but instead of removing a solution from the queue every $\bbam$ steps, we use a well chosen value, derived from the amortized delay up to the current time (see Line~\ref{line:ceps}).

  First observe that during the execution of the algorithm, we always have $d \leq p$. Indeed, $d$ is by definition always smaller than $\max_t {t \over N_t+1}$ where $N_t$ is the number of solutions output at time $t$ and hence cannot exceed the amortized delay $p$ of $\abbox$. Moreover, between the output of two solutions, the while loop of \Cref{alg:enumqueue_improved} is executed at most $\lceil d \log(d)^2 \rceil \leq \lceil p \log(p)^2 \rceil$ times. Hence, if proven correct, the delay of \Cref{alg:enumqueue_improved} is at most $O(wp\log(p)^2)$ where $w$ is the complexity of executing the while loop.

  There are two possible non constant operations in the while loop. The first is the computation of the ratio at Line~\ref{line:ceps}. We cannot compute it in constant time since both $S$ and $\steps(M)$ can become too large. Both $S$ and $\steps(M)$ are actually counters that we only need to increment by $1$. We can use a data structure based on Gray codes to encode their values which allows us to efficiently increment them. Moreover, by keeping track of the most significant bit of each counter, we can compute a sufficiently good upperbound of the ratio in constant time. The details are given in \Cref{sec:countersknown}. In this proof, we hence assume that  Line~\ref{line:ceps} is computed in constant time. 

  The second non constant operation is the insertion in the queue which may cost $O(b)$. Hence, in the worst case, the delay of \Cref{alg:enumqueue_improved} is at most $O(bp\log(p)^2)$. However, if for every $i$, the time between the output of solutions $s_i$ and $s_{i+1}$ is at least $|s_i|$, then we can amortize the insertion of $s_i$ in the queue over the $|s_i|$ next iteration of the loop. Hence, in this case, the worst case delay of \Cref{alg:enumqueue_improved} is $O(p\log(p)^2)$. 

It remains to show that each time a solution is pulled from the queue on Line~\ref{line:pull2}, the queue cannot be empty.  The integer $S$ is by construction the number of solutions found so far in the simulation of $\abbox$ and $j$ is the number of steps since the last output, when a solution is removed.  Let $Z_0 = [0,1]$ and $Z_l = [2^l,2^{l+1}[$. We let $d_l$ be the value of $d$ in Algorithm~\ref{alg:enumqueue_improved}, when the simulation $M$ of $\abbox$ is at the beginning of $Z_l$, that is, when $\bbsteps = 2^l$.
%
% We prove that when Algorithm~\ref{alg:enumqueue_improved} outputs a solution when $\bbsteps \in Z_l$, there are still solutions in the queue.
We let $S_{l}$ be the number of solutions output by $\abbox$ before time $2^l$ and let $P_l$ be the number of solutions output by \Cref{alg:enumqueue_improved} when  $ \bbsteps = 2^{l+1}$. We prove that for every $l$, $P_l \leq S_{l}$. In other words, when $\bbsteps \in Z_l$, we have found at least $S_l$ solution and output at most $P_l$, hence, we never pull from an empty queue. 

  The value $P_l$ is the sum of solutions output over each slice $Z_i$ up to $Z_l$. Since $d$ is increasing in \Cref{alg:enumqueue_improved}, the number of solutions output while $\bbsteps \in Z_i$ is at most ${|Z_i| \over d_i\log(d_i)^2}$ since $d_i$ is the value of $d$ when $\bbsteps = 2^i$, that is, when $M$ is at the beginning of $Z_i$. Hence, $$P_l \leq \sum_{i=0}^{l} \frac{|Z_i|} {d_i\log(d_i)^2}.$$ 

   Since $d$ is initialized to $8$ and is increasing, for all $i$ we have $d_i\log(d_i)^2 \geq 64$
   and we use this fact to bound the first terms of the previous sum.

   \begin{align*}
     \sum_{i=0}^{l-\lfloor \log(d_l) \rfloor + 1} \frac{|Z_i|} {d_i\log(d_i)^2}
     & \leq \sum_{i=0}^{l-\lfloor \log(d_l) \rfloor + 1} \frac{|Z_i|}{64} \\ 
     & \leq \frac{1}{64} \sum_{i=0}^{l- \lfloor \log(d_l) \rfloor + 1} 2^{i} , \text{ since } |Z_i|=2^i \\
     & \leq \frac{1}{64} 2^{l-\lfloor \log(d_l) \rfloor +2} \leq \frac{2^{l}}{8d_l} \\
     & \leq \frac{S_l+1}{8}.
   \end{align*}
   The last inequality follows from the fact that $d_l \geq {2^l \over S_l +1 }$. Indeed, $d_l$ and $S_l$ are respectively the values of $d$ and $S$ in \Cref{alg:enumqueue_improved} when $\bbsteps = 2^l$. 

   We now proceed to bound the last terms of the sum. Observe that, $d_{i+1} \leq  2d_{i}$. Indeed, by definition, $d_i \geq {2^i \over S_i+1}$ and $d_{i+1} = \max_{2^i < j \leq 2^{i+1}} {j \over S_i+k_j+1}$ where $k_j$ is the number of solutions found by $\abbox$ between time $2^i$ and $2^{i}+j$. Hence $d_{i+1} \leq {2^{i+1} \over S_i+1}$ by upper bounding $j$ with $2^{i+1}$ and lower bounding $k_j$ by $0$. That is $d_{i+1} \leq 2d_i$ and by a straightforward induction, $d_{l-i} \leq 2^{-i}d_{l}$ for every $i$.

   % Hence, for all $ i \leq l$, $d_{l-i} \geq 2^{-i}d_{l}$. By definition, $d_l$ is always smaller than $\max(4,2^l)$, thus $\log(d_l) \leq \max(2,l) \leq l+2$.
   % Hence, $Z_{l-i}$ is defined for $i \leq \log(d_l)-2$ and we can write the following equation.

   \begin{align*}
     \label{eq:adaptative-large}
    \sum_{i=l-\lfloor \log(d_l) \rfloor +2}^{l} {|Z_{i}| \over d_{i}\log(d_{i})^2} & =  \sum_{j=0}^{\lfloor \log(d_l) \rfloor -2} {|Z_{l-j}| \over d_{l-j}\log(d_{l-j})^2} \\
     & \leq \sum_{j=0}^{\lfloor \log(d_l) \rfloor-2} {2^{l-j} \over 2^{-j}d_l\log(2^{-j}d_l)^2} \text{ from what preceeds} \\
     & \leq {2^{l} \over d_l} \sum_{j=0}^{\lfloor \log(d_l) \rfloor-2} {1 \over (\log(d_l) - j)^2} \\
                                                                                   & \leq {2^{l} \over d_l} \sum_{j=0}^{\lfloor \log(d_l) \rfloor-2} {1 \over (\lfloor \log(d_l) \rfloor - j)^2} 
                                                                                    = {2^{l} \over d_l} \sum_{i=2}^{\lfloor \log(d_l) \rfloor} {1 \over i^2} \\
                                                                                   & \leq {2^{l} \over d_l} ({\pi^2 \over 6}-1) \\
    & \leq {3 \over 4}{2^{l} \over d_l} \leq {3\over 4} (S_l+1) 
  \end{align*}

By summing both inequalities, we obtain $P_l \leq {7 \over 8}(S_{l}+1)$. Since both $P_l$ and $S_l$ are integers, we conclude that $P_l \leq S_l$ which proves the proposition.
\end{proof} 

As before, by modifying the blackbox oracle using \Cref{alg:pacing}, we immediately get:
\begin{corollary}\label{cor:adaptative_best}
  There exists a regularization scheme, independant from $\bbam$, such that for every enumeration process $\abbox$, the worst case delay of the scheme is $O(p \log(p)^2)$ where $p = \bbsize + \amdelay[\abbox]$.
\end{corollary}

\Cref{cor:adaptative_best} adapts to the unknown amortized delay of the algorithm it regularizes. One can however observe that the worst case delay in this case is not as good as the one of \Cref{prop:queue}, which is linear. We formally show in \Cref{sec:lower_bound_unknown_delay} (\Cref{thm:lower_bound_udelay}) using an adversarial blackbox $\abbox$ that there is no regularization scheme independant from $\bbam$ with $O(\amdelay[\abbox])$ worst case delay.

We conclude this section by observing that in the proof of \Cref{cor:adaptative_best}, the second inequality works because $\sum_{i=1}^\infty {1\over i^2}$ is bounded. If we replace the $d\log(d)^2$ bound of Line~\ref{line:ceps}  in \Cref{alg:enumqueue_improved} by $d \log(d)^{1+\varepsilon}$ for $\epsilon > 0$, the same proof as \Cref{cor:adaptative_best} works because $\sum_{i=2}^\infty {1\over i^{1+\varepsilon}} < +\infty$. Hence, if the initial value of $d$ is large enough, we would also prove the inequality $P_l \leq S_l$ in the same way and get:
\begin{theorem}\label{thm:adaptative-eps}
  For every $\varepsilon>0$, there exists a regularization scheme, independant from $\bbam$, such that for every enumeration process $\abbox$, the worst case delay of the scheme is $O(p \log(p)^{1+\varepsilon})$ where $p = \bbsize + \amdelay[\abbox]$.
\end{theorem}

\subsection{Amortizing Self-Reducible Problems}
\label{sec:amort-self-reduc}

In this section, we show how the classical buffer amortization method (Algorithm~\ref{alg:enumqueue}) can be implemented with only a polynomial space overhead \emph{on a certain class of enumeration algorithms}. Therefore, the results of this section are not expressed in the framework of enumeration process, contrary to the rest of the article. We use this method to prove that the worst case delay of these algorithms can be made almost as good as their \emph{average delay} and not only their \emph{amortized delay}. To obtain this result in the conference version of this article~\cite{CapelliS23}, we did use the more elaborate geometric regularization presented in the next section. In a private communication, Takeaki Uno pointed out that a method similar to ours using a bounded buffer achieves the same result. This method was presented by Uno under the name \textbf{Output Queue Method} in an unpublished technical report~\cite{uno03}. 

Given an enumeration problem $\enum{A}$, we assume in this section that the solutions in $A(x)$ are sets over some universe $U(x)$. From $A$, we define the predicate $\tilde{A}$ which contains the pairs $((x,a,b),y)$ such that $y \in A(x)$ and $a \subseteq y \subseteq b$.
 Moreover, we define a self-reducible\footnote{For a classical definition of self-reducible problems, see~\cite{khuller1991planar,balcazar1990self}.} variant of $\enum{A}$ and the extension problem $\ext{A}$ defined as the set of triples $(x,a,b)$ such that there is a $y$ in $\tilde{A}(x,a,b)$.

Solving $\enum{A}$ on input $x$ is equivalent to solving $\enum{\tilde{A}}$ on $(x,\emptyset,U(x))$. Let us now formalize a recursive method to solve $\enum{\tilde{A}}$, sometimes called \emph{binary partition}~\cite{uno1998new}, because it partitions the solutions to enumerate in two disjoint sets. Alternatively, it is called backtrack~\cite{read1975bounds} or \emph{flashlight search}~\cite{boros_generating_2009}, because we peek at subproblems to solve them only if they yield solutions.  To our knowledge, all uses of flashlight search in the literature can be captured by this formalization, except for the partition of the set of solutions which can be in more than two subsets. We only present the binary partition for the sake of clarity, but our analysis can be adapted to finer partitions.

Given an instance $(x,a,b)$ of $\enum{\tilde{A}}$ and some global auxiliary data $D$, a flashlight search consists in the following (subroutines are not specified, and yield different flashlight searches):
\begin{itemize}
    \item if $a=b$, $a$ is ouput 
    \item otherwise choose $u \in b \setminus a$; 
    \begin{itemize}
        \item if $(x,a\cup\{u\},b) \in \ext{A}$, compute some auxiliary data $D_1$ from $D$ and
        make a recursive call on $(x,a\cup\{u\},b)$;
        \item if $(x,a,b \setminus \{u\}) \in \ext{A}$, compute some auxiliary data $D_2$ from $D_1$ and make a recursive call on $(x,a,b \setminus \{u\})$, then compute $D$ from $D_2$.
    \end{itemize}
\end{itemize}

Flashlight search can be seen as a depth-first traversal of a \emph{partial solutions tree}. A node of this tree is a pair $(a,b)$ such that $(x,a,b) \in \ext{A}$. Node $(a,b)$ has children $(a\cup\{u\},b)$ and $(a,b \setminus \{u\})$ if they are nodes. A leaf is a pair $(a,a)$ and the root is $(\emptyset,U(x))$. The \emph{cost} of a node $(a,b)$ is the time to execute the flashlight search on $(x,a,b)$ \emph{except the time spent in recursive calls}. Usually, the cost of a node comes from deciding $\ext{A}$ and modifying the global data structure $D$ used to solve $\ext{A}$ faster. 

The cost of a path in a partial solution tree is the sum of the costs of the nodes in the path. The \emph{path time} of a flashlight search algorithm is defined as the maximum over the cost of all paths from the root. The delay is bounded by twice the path time since, between two output solutions, a flashlight search traverses at most two paths in the tree of partial solutions. To our knowledge, all bounds on the worst case delay of flashlight search are proved by bounding the path time. The path time is bounded by $\sharp U(x)$ times the complexity of solving $\ext{A}$ plus the complexity of updating the auxiliary data. Auxiliary data can be used to amortize the cost of evaluating $\ext{A}$ repeatedly, generally to prove that the path time is equal to the complexity of solving $\ext{A}$ once, e.g., when generating minimal models of monotone CNF~\cite{murakami2014efficient}.

Using flashlight search, we obtain that if $\ext{A} \in \P$ then $\enum{A} \in \DelayP$. Many enumeration problems have actually been shown to be in $\DelayP$ by showing their extension problem is in $\P$~\cite{uno2015constant,MaryS19}. However, some $\DelayP$ problems provably escape this framework, e.g., the extension of a maximal clique which is in $\DelayP$ but for which hardness of the extension problem can be derived from the fact that finding the largest maximal clique in lexicographic order is \NP-hard~\cite{johnson1988generating}.

The average delay or the worst case delay of a flashlight search is given for any input $(x,a,b)$ of $\enum{\tilde{A}}$ and not only for the input $(x,\emptyset,U(x))$ corresponding to an input of $\enum{A}$. We make the assumption that $|U(x)|$ is less than $n = |x|$, hence an instance $(x,a,b)$ is of size less than $3n$ and we give all delays as functions of $n$. The average delay of flashlight search is sometimes much smaller than its worst case delay, especially when the internal nodes of the tree of partial solutions are guaranteed to have many leaves. Uno describes the pushout method~\cite{uno2015constant} harnessing this property to obtain constant average delay algorithms for many problems such as generating spanning trees. 

We now prove that the average delay of a flashlight search can be turned into a worst case delay, by using a buffer of bounded size. The only space overhead is the size of this buffer and there is no time overhead. However, contrarily to all regularization schemes presented in this article, there is an increase in the preprocessing which becomes larger than the obtained worst case delay.

 We present a proof similar to the one in~\cite{uno03}, which is based on the fact that we can bound the average delay of a flashlight search between any two points of time, a property even stronger than having a polynomial amortized delay.

\begin{figure}
  \begin{center}
    \def\svgwidth{6cm} 
  %% Creator: Inkscape 1.2.2 (b0a8486541, 2022-12-01), www.inkscape.org
%% Accompanies image file 'outputqueue.pdf' (pdf, eps, ps)
\begingroup%
  \makeatletter%
  \providecommand\color[2][]{%
    \errmessage{(Inkscape) Color is used for the text in Inkscape, but the package 'color.sty' is not loaded}%
    \renewcommand\color[2][]{}%
  }%
  \providecommand\transparent[1]{%
    \errmessage{(Inkscape) Transparency is used (non-zero) for the text in Inkscape, but the package 'transparent.sty' is not loaded}%
    \renewcommand\transparent[1]{}%
  }%
  \providecommand\rotatebox[2]{#2}%
  \newcommand*\fsize{\dimexpr\f@size pt\relax}%
  \newcommand*\lineheight[1]{\fontsize{\fsize}{#1\fsize}\selectfont}%
  \ifx\svgwidth\undefined%
    \setlength{\unitlength}{340.19130436bp}%
    \ifx\svgscale\undefined%
      \relax%
    \else%
      \setlength{\unitlength}{\unitlength * \real{\svgscale}}%
    \fi%
  \else%
    \setlength{\unitlength}{\svgwidth}%
  \fi%
  \global\let\svgwidth\undefined%
  \global\let\svgscale\undefined%
  \makeatother%
  \begin{picture}(1,1.14750807)%
    \lineheight{1}%
    \setlength\tabcolsep{0pt}%
    \put(0,0){\includegraphics[width=\unitlength,page=1]{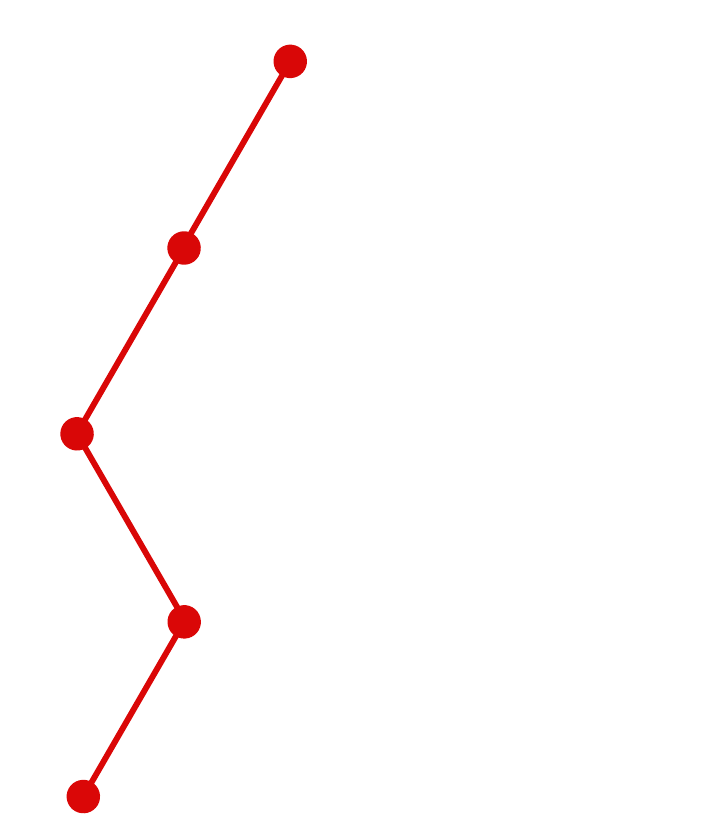}}%
    \put(0.02384496,0.01150413){\color[rgb]{0,0,0}\makebox(0,0)[lt]{\lineheight{1.25}\smash{\begin{tabular}[t]{l}$u$\end{tabular}}}}%
    \put(0,0){\includegraphics[width=\unitlength,page=2]{outputqueue.pdf}}%
    \put(0.3937266,1.09954965){\color[rgb]{0,0,0}\makebox(0,0)[lt]{\lineheight{1.25}\smash{\begin{tabular}[t]{l}$w$\end{tabular}}}}%
    \put(0.91138715,0.25015151){\color[rgb]{0,0,0}\makebox(0,0)[lt]{\lineheight{1.25}\smash{\begin{tabular}[t]{l}$v$\end{tabular}}}}%
  \end{picture}%
\endgroup%

\caption{A traversal of the tree of partial solutions by the flashlight search. We represent the subtree visited between time $t$ when the algorithm visits
$u$ and time $t'$ when the algorithm visits $v$. The subproblems completely solved recursively are in blue, the path between $u$ and $v$ in red.\label{fig:outputqueue}}
\end{center}
\end{figure}

\begin{lemma}\label{lemma:flashlightdelay}
Let $\enum{A}$ be an enumeration problem solved by a flashlight search algorithm $I$
with path time $p(n)$ and average delay $d(n)$. Then, between times $t$ and $t'$, with $t<t'$, algorithm $I$ outputs at least $\lceil{t'- t - 2p(n) \over d(n)}\rceil$ solutions.
\end{lemma}
\begin{proof}
Let us denote by $u$ the node of the partial solution tree visited by $I$ at time $t$ and by $v$ the node visited at time $t'$. Let us consider the path between $u$ and $v$ and let $w$ be their least common ancestor as represented in Figure~\ref{fig:outputqueue}. 
Algorithm $I$ is a flashlight, thus it does a depth first traversal of the partial solution tree. Hence, the time spent by $I$ from $t$ to $t'$ is bounded by:

 % \flo{ The $u \rightarrow^* w$ part is not really paid but for the ``backtracking'' since we are already down in $u$ at time $t$? Yann: It is paid, since the only cost of backtracking is updating the data back to its previous form, which is paid for in the path time. Do not think it is worth to mention.}
\begin{itemize}
    \item the time spent in the path from $u$ to $w$ bounded by $p(n)$ the path time
    and the time in the path from $w$ to $v$ also bounded by $p(n)$ (red paths in Figure~\ref{fig:outputqueue})
    \item the time spent in suproblems which are right children of elements in the path from $u$ to $w$ or left children of elements in the path from $w$ tu $v$
    (blue triangles in Figure~\ref{fig:outputqueue})
\end{itemize}

Therefore, the time spent by $I$ in subproblems along the path between $u$ and $v$ is at least $t'- t - 2p(n)$. Each subproblem is an instance of $\enum{\tilde{A}}$ of the same size $n$, hence $I$ recursively solves it with average delay bounded by $d(n)$. Therefore, these subproblems contributes at least $\lceil(t'- t - 2p(n))/d(n)\rceil$ solutions, which are produced between time $t$ and $t'$.
\end{proof}

We use \Cref{lemma:flashlightdelay} to show that a slight modification of \Cref{alg:enumqueue} where the queue size is bounded can be used to improve the worst case delay of a flashlight search algorithm for $\enum{A}$:

\begin{theorem}\label{th:self-reduction}
Let $\enum{A}$ be an enumeration problem solved by a flashlight search algorithm, with space $s(n)$, path time $p(n)$, average delay $d(n)$ and $b(n)$ the size of a single solution. There is an algorithm solving $\enum{A}$ on any input $x$, with preprocessing $O(p(n))$, delay $O(d(n)+ b(n))$ and space $O(s(n) + \frac{p(n)b(n)}{d(n)+b(n)})$.
\end{theorem}
\begin{proof}
Let $I$ be the flashlight search algorithm solving $\enum{A}$. We apply the tranformation of \Cref{alg:pacing} to $I$ to obtain $I'$, a flash light search whith an average delay of $D(n) \in O(d(n) + b(n))$ and the same path time $p(n)$. This transformation ensures that Algorithm~\ref{alg:enumqueue} takes constant time per iteration on average between two outputs.

We consider a slightly modified version of Algorithm~\ref{alg:enumqueue} where the queue $Q$ is of bounded size: the queue cannot contain more than $\lceil 2p(n)/D(n) \rceil + 1$ solutions. When an element is inserted into a full queue, the first element of the queue is output to make room for the new one. Finally, the regular output of solutions from $Q$ is done every $D(n)$ simulation steps and begins after the queue is full for the first time.

Let $I''$ be the algorithm obtained by the regularization of $I'$ with a bounded queue as described in the previous paragraph. The preprocessing step of $I''$ is to simulate $I'$ up to the point $Q$ is full. We apply Lemma~\ref{lemma:flashlightdelay} with $t=0$ and $t'= 4p(n) + 1$, which ensures that the number of solutions inserted into $Q$ is at least $\lceil 2p(n)/D(n) \rceil + 1$. Therefore, $I''$ do at most $4p(n) + 1$ simulations steps of $I'$.

By construction of the algorithm, each $D(n)$ simulation steps or when the queue is full and a solution is inserted, the first solution of the queue is output. Hence, the delay of the algorithm is bounded by $O(D(n))$.

We only need to prove that the queue is never emtpy during the algorithm, when we remove a solution. Let us consider $t_1$ any point in the runing time of $I'$ and let us prove that the queue contains at least a solution when $I''$ simulates the $t_1$th time step of $I'$. Let $t_0$ be the largest time smaller than $t_1$ such that $Q$ is full. This time exists, since the algorithm fills up the queue completely during the preprocessing. By Lemma~\ref{lemma:flashlightdelay}, applied to $t_0$ and $t_1$, at least $\lceil{ t_1- t_0 - 2p(n) \over D(n)} \rceil$ solutions have been inserted between $t$ and $t'$. Moreover, solutions have been output exactly every $D(n)$ steps, since the queue was never full inbetween $t_0$ and $t_1$. Therefore, at most $\lceil (t_1-t_0)/D(n)\rceil$ solutions were removed from the queue. Thus, the number of solutions in $Q$ has decreased by at most $\lceil 2p(n)/D(n) \rceil$ between $t_0$ and $t_1$. Since $Q$ contains $\lceil 2p(n)/D(n) \rceil + 1$ solution at time $t_0$, at time $t_1$ the queue contains at least one solution.
\end{proof}

 Theorem~\ref{th:self-reduction} has a direct application to the problem of generating models of a DNF formula, which has been extensively studied by the authors in~\cite{capelli2021enumerating}. Let us denote by $n$ the number of variables of a DNF formula, by $m$ its number of terms and by $\enum{DNF}$ the problem of generating the models of a DNF formula. The size of a DNF formula is at least $m$ and at most $O(mn)$ (depending on the representation and the size of the terms), which can be exponential in $n$. Hence, we want to understand whether $\enum{DNF}$ can be solved with a delay polynomial in $n$ only, that is depending on the size of a model of the DNF formula but not on the size of the formula itself. A problem that admits an algorithm with a delay polynomial in the size of a single solution is said to be \emph{strongly polynomial} and is in the class $\SDelayP$. One typical obstacle to being in $\SDelayP$ is dealing with large non-disjoint unions of solutions. The problem $\enum{DNF}$ is an example of such difficulty: its models are the union of the models of its terms, which are easy to generate with constant delay. 

The paper~\cite{capelli2021enumerating} defines the \emph{strong DNF enumeration conjecture} as follows: there is no algorithm solving $\enum{DNF}$ in delay $o(m)$. It also describes an algorithm solving $\enum{DNF}$ in \emph{average} sublinear delay. It is based on a tight analysis of a flashlight search, with appropriate data structures and choice of variables to branch on (Theorem $10$ in~\cite{capelli2021enumerating}). Thanks to Theorem~\ref{th:self-reduction}, we can trade the average delay for a worst case delay and falsify the strong DNF enumeration conjecture.

\begin{corollary}\label{cor:DNF}
There is an algorithm solving $\enum{DNF}$ with preprocessing $O(mn)$, delay $O(m^{1-\log_3(2)}n)$ and space $O(mn)$.
\end{corollary}
\begin{proof}
 The algorithm of~\cite{capelli2021enumerating} enumerates all models with average delay $O(m^{1-\log_3(2)}n)$. The space used is the representation of the DNF formula by a trie, that is $O(mn)$. The path time is the cost of modifying the trie along a branch of the flashlight, that is $O(mn)$. The size of a solution is $n$. We apply Theorem~\ref{th:self-reduction} to this algorithm to prove the result.
\end{proof}

For monotone DNF formulas, Theorem~$13$ of~\cite{capelli2021enumerating} gives a flashlight search with an average delay $O(n)$. Hence, we obtain an algorithm with delay $O(n)$ listing the models of monotone DNF formulas by Theorem~\ref{th:self-reduction}. It gives an algorithm having an exponentially better delay, preprocessing and space usage than the algorithm given by Theorem~$12$ of~\cite{capelli2021enumerating}. 

\begin{corollary}\label{cor:monotone-DNF}
There is an algorithm solving $\enum{DNF}$ on monotone formulas with preprocessing $O(mn)$, delay $O(n)$ and space $O(mn)$.
\end{corollary}
%%% Local Variables:
%%% mode: latex
%%% TeX-master: "theoretics"
%%% End:

\section{Geometric Regularization}
\label{sec:mainalgo}

\subsection{Geometric Regularization}
\label{sec:geoamo}

\Cref{alg:enumqueue} transforms the amortized delay of an enumeration process $\abbox$ into guaranteed worst case delay. However, it uses a queue to delay the output of solutions which may become large, possibly linear in the number of solutions of $\abbox$. When applied to a problem in $\AmDelayP$, it means that \Cref{alg:enumqueue} can use a queue of size exponential in the input size. A natural question is hence to understand whether one can design a regularization scheme which has delay polynomial in the amortized delay and which space overhead is better than linear in the number of solutions.

We answer this question positively and introduce a new technique called geometric regularization and whose pseudocode is given in \Cref{alg:enumit}. We summarize the guarantees of this new regularization scheme in the theorem below:
% To overcome this issue, we introduce a technique that we call \emph{geometric amortization}, illustrated by Algorithm~\ref{alg:enumit} which regularizes the delay of an $\Inc1$-enumerator with a space overhead of $O(\log(\sharp I(x)))$, which is polynomially bounded since $I$ is in $\EnumP$. To achieve this, we however have to compromise a bit on the delay which becomes $O((\log(\sharp I(x)) \cdot p(|x|))$. Moreover, with geometric amortization, the solutions are not output in the same order as the order they are output by $I$. Algorithm~\ref{alg:enumit} relies on the knowledge of an upper bound $K$ of $\sharp I(x)$, but this assumption is relaxed in Section~\ref{sec:solution_unknown}. We now proceed to prove the correctness and complexity of Algorithm~\ref{alg:enumit} that is summarized in the theorem below.

\SetKw{Break}{break}
\SetKwInOut{KwInput}{Input}
\SetKwInOut{KwOutput}{Output}

\begin{algorithm}
  \KwInput{Enumeration process $\abbox$} % $x \in \Sigma^*$ of size $n$ and $K$ such that $K \geq \sharp I(x)$}
  \KwOutput{Enumerate $\abbox$ with delay bounded by $O(\bbam \cdot \log(\bbsol))$ and space $O(\log(\bbsol) \cdot \bbspace)$}
  \Begin{
    $N \gets \lceil \log(\bbsol) \rceil$, $d \gets \bbam$\;
    \lFor{$i=0$  \KwTo $N$}{
      $M[i] \gets \bbload$
    }
    $j \gets N$\;
    \While{$j \geq 0$}{\label{enumit:line:while}
      \For{$B \gets 2 \cdot d$  \KwTo $0$}{\label{enumit:line:for}
        $\ramstep(M[j])$\;
        \If{$\bbisoutput[M[j]]$ \textbf{\emph{and}} $\bbsteps[M[j]]
          \in Z_j$ \label{line:test}}{ 
          $\out(\bboutput[M[j]])$\;
          $j \gets N$\;
          \Break;
        }
      }
      \lIf{$B = 0$}{$j \gets j-1$}
    }
  }
  \caption{Geometric regularization, guarantees are in \Cref{th:geometric}. In the code, $d$ stands for $\bbam$ and $Z_0=[0;d]$ and $Z_j = [2^{j-1}d+1;2^jd]$ for $j>0$.}
  \label{alg:enumit}
\end{algorithm}

% \subsubsection{Correctness and Complexity of Algorithm~\ref{alg:enumit}}
% \label{sec:basic}

\begin{theorem}
  \label{th:geometric} For every enumeration process $\abbox$, \Cref{alg:enumit} with oracle $\abbox$ is an enumeration scheme with worst case delay $O(\bbam \cdot \log(\bbsol))$ and space $O(\log(\bbsol))$.

%  Given an oracle enumerator $\abbox$ with incremental delay $p(n)$ and space complexity $s(n)$ and given $K \geq \sharp I(x)$, one can construct a $\DelayP$-enumerator $I'$ which enumerates $I(x)$ on any input $x \in \Sigma^*$ with delay $O(\log(K)p(n))$ and space complexity $O(s(n)\log(K))$. 

\end{theorem}

Before presenting the algorithm in details and proving \Cref{th:geometric}, we explain why it has interesting consequences on the complexity of enumeration problem. While \Cref{thm:delayp_is_amdelayp} establishes that $\AmDelayP = \DelayP$, it is not sufficient to prove that the class of problems with polynomial amortized delay and polynomial space, denoted by $\AmDelayP^\poly$, is the same as the class of problems with polynomial delay and polynomial space, denoted by $\DelayP^\poly$. The regularization scheme from \Cref{alg:enumit} allows us to prove the more general result:

\begin{theorem}
  \label{thm:amdelayp_is_amdelayp_poly}
  $\AmDelayP^{\poly} = \DelayP^\poly$.
\end{theorem}
\begin{proof}
  Let $A$ be a problem in $\AmDelayP^{\poly}$. By definition, there is a polynomial $b$ such that for every $x$, the solutions $y \in A(x)$ have size at most $b(|x|)$. In particular, it means that the size of $A(x)$ is at most $2^{b(|x|)}$. Moreover, there exists an algorithm $I$ for $A$ such that on input $x$, it outputs $A(x)$ with amortized delay $d(|x|)$ and space $s(|x|)$, where both $s$ and $d$ are polynomials.

  We now design a polynomial delay algorithm for $A$ as follows: on input $x$, we run the regularization scheme from \Cref{alg:enumit} with oracle $\abbox := I(x)$, $\bbam := d(|x|)$ and $\bbsol := 2^{b(|x|)}$. By \Cref{th:geometric}, this algorithm outputs $A(x)$ with delay $O(d(|x|)b(|x|))$ and space $O(b(|x|))$ which are both polynomial in $|x|$. Now \Cref{alg:enumit} uses RAM machine with blackbox access. However, one can emulate the blackbox on a true RAM machine where each operation takes constant time and each pointer to a machine uses only $O(s(|x|)+\log(d(|x|))^2)$ memory which is polynomial in $|x|$, see \Cref{sec:oracles} for more details. Hence, $A \in \DelayP^\poly$. 
\end{proof}

The rest of this section is dedicated to the proof of \Cref{th:geometric}. The idea of geometric regularization from \Cref{alg:enumit} is to simulate several copies of $\abbox$ at different speeds. Each copy is responsible for enumerating solutions in different intervals of time to avoid repetitions in the enumeration. The copy which are responsible for enumerating solutions that appear later in time intuitively move faster than the others. The name comes from the fact that the size of the intervals we use follows a geometric progression (the size of the $(i+1)^\th$ interval is twice the size of the $i^\th$ one).

\subparagraph{Explanation of Algorithm~\ref{alg:enumit}.} In this section, to ease notation, we denote by $d := \bbam$. Algorithm~\ref{alg:enumit} maintains $N+1$ simulations $M[0], \dots, M[N]$ of $\abbox$ where $N = \lceil \log(\bbsol) \rceil$. When simulation $M[i]$ finds a solution, it outputs it if and only if the number of steps of $M[i]$ is in $Z_i$, where $Z_i := [1+2^{i-1} d, 2^i  d]$ for $i>0$ and $Z_0 = [1,d]$. These intervals are clearly disjoint and cover every possible step of the simulation since the total time of $\abbox$ is at most $\sharp \abbox \times \amdelay[\abbox] \leq \bbsol \times d \leq 2^N d$ (by convention, we assumed a enumeration processes stops on its last solution, see Section~\ref{sec:preliminaries}). Thus, every solution is enumerated as long as every $M[i]$ has reached the end of $Z_i$ when the algorithm stops.

Algorithm~\ref{alg:enumit} starts by moving $M[N]$. It is given a budget of $2d$ steps. If these $2d$ steps are executed without finding a solution in $Z_N$, $M[N-1]$ is then moved similarly with a budget of $2d$ steps. It continues until one machine $M[j]$ finds a solution in its zone $Z_j$. In this case, the solution is output and the algorithm proceeds back with $M[N]$. The algorithm stops when $M[0]$ has left $Z_0$, that is when $d+1$ steps of $M[0]$ have been simulated\footnote{An illustration of Algorithm~\ref{alg:enumit} can be found at \url{http://florent.capelli.me/coussinet/} where one can see the run of a machine represented as a list and the different simulations moving in this list and discovering solutions.}.
 
\subparagraph{Bounding the delay.} From the above description of Algorithm~\ref{alg:enumit}, between two outputs, the while loop at Line~\ref{enumit:line:while} is executed at most $N+1$ times and the for loop at Line~\ref{enumit:line:for} is executed at most $2d+1$ times. Instructions $\bbstep[\cdot]$, $\bbisoutput[\cdot]$ and $\bboutput[\cdot]$ are executed in constant time (see \Cref{sec:preliminaries}). We also assume that the condition $\bbsteps[M[j]] \in Z_j$ can be tested in constant time. Similarly as in \Cref{alg:enumqueue}, counter $\bbsteps[M[j]]$ may become too big for this test to be really performed in $O(1)$ but using Gray code counter and keeping track of the most significant bit is enough in this case, see \Cref{sec:counters} for details. Hence, the body of the for loop can be executed in constant time. Hence, the time between two outputs is at most $O(d \cdot N) = O(\bbam \cdot \log(\bbsol))$. 

\subparagraph{Space complexity.}  % We have seen in Section~\ref{sec:preliminaries} that each enumeration process can be simulated without using more space than the original machine (see \Cref{sec:oracles} for more details). 
 Algorithm~\ref{alg:enumit} uses  $N = O(\log(\bbsol))$ simulations of $\abbox$, and
each simulation uses a space $O(1)$ in the blackbox model. It also uses a constant number of variables and there values are bounded by $2\bbam$ and  $\lceil \log(\bbsol) \rceil$ which can be considered as inputs. Hence, these variables can be stored in constant space. The space complexity is made precise in \Cref{sec:oracles}, where we have a RAM running on an input instead of an enumeration process.

\subparagraph{Correctness of Algorithm~\ref{alg:enumit}.} It remains to show that Algorithm~\ref{alg:enumit} correctly outputs every solution of $\abbox$ before stopping. Recall that a solution of $\abbox$ is enumerated by $M[i]$ if it is produced by $\abbox$ at step $c \in Z_i = [1+2^{i-1}d,2^id]$.  Since, by definition, the total time of $\abbox$ is at most $\bbsol \cdot d$, it is clear that $Z_0 \uplus \dots \uplus Z_{N}  \supseteq [1, \bbsol d]$ covers every solution and that each solution is produced at most once. Thus, it remains to show that when the algorithm stops, $M[i]$ has moved by at least $2^id$ steps, that is, it has reached the end of $Z_i$ and output all solutions in this zone.

% Todo
We study the execution of Algorithm~\ref{alg:enumit}. For the purpose of the proof, we only need to look at the values of $\bbsteps[M[0]],\dots,\bbsteps[M[N]]$ during the execution of the algorithm. We thus say that the algorithm is in state $c = (c_0,\dots,c_N)$ if $\bbsteps[M[i]]=c_i$ for all $0 \leq i\leq N$. We denote by $S^c_i$ the set of solutions that have been output by $M[0], \dots, M[i]$ when state $c$ is reached; that is, a solution is in $S^c_i$ if and only if it is produced by $\abbox$ at step $k \in Z_j$ for $j \leq i$ and $k \leq c_j = \bbsteps[M[j]]$. We claim the following invariant: 
\begin{lemma}
  \label{lem:c_vs_si} For every state $c$ reached by \Cref{alg:enumit} and $i < N$, we have $c_{i+1} \geq 2d|S^c_i|$.
\end{lemma}
\begin{proof}
  The proof is by induction on $c$. For the state $c$ just after initializing the variables, we have that for every $i \leq N$, $|S^{c}_i| = 0$  and $c_i = 0$. Hence, for $i<N$, $c_{i+1} \geq 0 = 2d|S^{c}_i|$.

  Now assume the statement holds at state $c'$ and let $c$ be the next state. Let $i<N$. If $|S_i^c| = |S_i^{c'}|$, then the inequality still holds since $c_{i+1} \geq c'_{i+1}$ and $c'_{i+1}\geq 2d|S^{c'}_i|=2d|S^{c}_i|$ by induction. Otherwise, we have $|S_i^c| =|S_i^{c'}|+1$, that is, some simulation $M[k]$ with $k \leq i$ has just output a solution. In particular, the variable $j$ has value $k$ with $k \leq i < N$. Let $c''$ be the first state before $c'$ such that variable $j$ has value $i+1$ and $b$ has value $2d$, that is, $c''$ is the state just before Algorithm~\ref{alg:enumit} starts the for loop to move $M[i+1]$ by $2d$ steps.  No solution has been output between state $c''$ and $c'$ since otherwise $j$ would have been reset to $N$. Thus, $|S_i^{c''}| = |S_i^{c'}|$. Moreover, $c_{i+1} \geq c'_{i+1} \geq c''_{i+1}+2d$ since $M[i+1]$ has moved by $2d$ steps in the for loop without finding a solution.  By induction, we have $c''_{i+1} \geq 2d|S_i^{c''}| = 2d|S_i^{c'}|$. Thus $c_{i+1} \geq c''_{i+1}+2d = 2d(|S_i^{c'}|+1) = 2d|S_i^c|$ which concludes the induction.
\end{proof}

\begin{corollary}
  \label{cor:enumitends} Let $c=(c_0,\dots,c_N)$ be the state reached when Algorithm~\ref{alg:enumit} stops. We have for every $i \leq N$, $c_i \geq 2^{i}d$.
\end{corollary}
\begin{proof}
  The proof is by induction on $i$. If $i$ is $0$, then we necessarily have $c_0 \geq d$ since Algorithm~\ref{alg:enumit} stops when the while loop is finished, that is, $M[0]$ has been moved by at least $2d$ steps.
  
  Now assume $c_j \geq 2^{j}d$ for every $j < i$. This means that for every $j<i$, $M[j]$ has been moved at least to the end of $Z_j$. Thus, $M[j]$ has found every solution in $Z_j$. Since it holds for every $j < i$, it means that $M[0], \dots, M[i-1]$ have found every solution in the interval $I=[1,2^{i-1}d]$. Since $\abbox$ has amortized delay $d$ and since $2^{i-1} \leq \bbsol$ by definition of $N$, interval $I$ contains at least $2^{i-1}$ solutions, that is, $|S^c_{i-1}| \geq 2^{i-1}$. Applying Lemma~\ref{lem:c_vs_si} gives that $c_i \geq 2^{i-1}\cdot 2d = 2^id$.
\end{proof}

The correctness of Algorithm~\ref{alg:enumit} directly follows from Corollary~\ref{cor:enumitends}. Indeed, it means that for every $i \leq N$, every solution of $Z_i=[1+2^{i-1}d, 2^id]$ have been output, that is, every solution of $[1, 2^{N}d]$ and $2^Nd$ is an upper bound on the total run time of $\abbox$.

Observe that Algorithm~\ref{alg:enumit} does not preserve the order of the solutions since it interleaves solutions later seen in the enumeration process in order to amortize large delay between two solutions. One possible workaround is to output pairs of the form $(s,r)$ where $s$ is a solution and $r$ its rank in the original enumeration order. To do so, one just has to keep a counter of the number of solutions seen so far by each process and output it along a solution. It does not restore the order, but allows recovering it afterward. 
When we are interested in finding the first $K$ solutions only (a likely scenario for solving top-k problems), the previous workaround is not useful. However, our algorithm can output them only with a $\log(K)$ overhead by just running $\log(K)$ processes and ignoring solutions having a rank higher than $K$; the solutions will however not be output in order. \Cref{thm:lower-bound-ordered} in \Cref{sec:lower-bounds} shows that this reordering is unavoidable. Indeed, in the blackbox oracle model, there are no regularization scheme preserving the order and whith a polynomial delay and space.

\subsection{Improving Algorithm~\ref{alg:enumit}} 

\label{sec:solution_unknown}
\begin{algorithm}
    \KwInput{Enumeration process $\abbox$} 
  \KwOutput{Enumerate $\abbox$ with delay bounded by $O(\bbam \cdot \log(\#\abbox))$ and space $O(\log(\#\abbox) \cdot \bbspace)$}
  \Begin{
    $d \gets \bbam$,
    $M \gets \mathrm{list}(\emptyset)$\;
    $\queueinsert(M,\bbload)$\;
    $j \gets \length(M)-1$\;
    
    \While{$j \geq 0$}{
        {
          \For{$B \gets 2 d$  \KwTo $0$}{
            $\bbstep[M[j]]$\;
            \If{$j=\length(M)-1$ \textbf{\em and} $\bbsteps[M[j]] = a_j$}{ \label{enumit_improved:line:newproc}
              $\queueinsert(M,\bbcopy[M[j]])$\; \label{enumit_improved:line:bbcopy}
              $j \gets \length(M)-1$\;
              \Break;
            }

            \If{$\bbisoutput[M[j]]$ \textbf{\emph{and}} $\bbsteps[M[j]] \in Z_j$}{
              $\out(\bboutput[M[j]])$\;
              $j \gets \length(M)-1$\;
              \Break;
            }
          }
          \lIf{$B = 0$}{$j \gets j-1$}
        }
      }
  }

  \caption{Improvement of Algorithm~\ref{alg:enumit} which is independant on the upper bound $\bbsol$ and has a better total time. In the code, $a_0 = 0$, $a_j=2^{j-1}d+1$ and $Z_j = [a_j;a_{j+1}[$ for $j>0$.}
  \label{alg:enumit_improved}
\end{algorithm}

One drawback of \Cref{alg:enumit} is that it needs to know in advance an upper bound $\bbsol$ on $\sharp \abbox$ since it uses it to determine how many simulations of $\abbox$ it has to maintain. As illustrated in \Cref{thm:amdelayp_is_amdelayp_poly}, it is sufficient to prove that the class of problems with amortized polynomial delay is the same as the class of problems with polynomial delay since this bound exists. However, in practice, it might be cumbersome to compute it or it may hurt efficiency if the upper bound is overestimated.
It turns out that one can remove this hypothesis by slightly modifying \Cref{alg:enumit}. The key observation is that during the execution of the algorithm, if $M[i]$ has not entered $Z_i$, it is simulated in the same way as $M[i+1], \dots, M[N]$. Indeed, it is not hard to see that $M[j]$ is always ahead of $M[i]$ for $j>i$ and that if $M[i]$ is not in $Z_i$, it will not output any solution in the loop at Line~\ref{enumit:line:for}, hence this iteration of the loop will move $M[i]$ by $2d$ steps, just like $M[j]$ for $j>i$. Hence, in theory, as long as $M[i]$ is not in $Z_i$, we could maintain only one simulation instead of $N-i$ and create a copy of $M[i]$ whenever it is needed.  \Cref{alg:enumit_improved} improves \Cref{alg:enumit} by implementing this strategy in the following way: we start with only two simulations $M[0], M[1]$ of $\abbox$. Whenever $M[1]$ is about to enter $Z_1$, we start $M[2]$ as an independent copy of $M[1]$. During the execution of the algorithm, we hence maintain a list $M$ of simulations of $\abbox$ and each time the last simulation $M[N]$ is about to enter $Z_N$, we copy it into a new simulation $M[N+1]$.

In practice, one has to be careful: one cannot really assume that copying simulation $M[N]$ can be done in constant time as the space it uses might be too big to copy. That can be achieved by lazily copying parts of $M[N]$ whenever we move $M[N+1]$. The details are given in \Cref{sec:copyorcales} and in this section, we will assume that the $\bbcopy$ operator used on Line~\ref{enumit_improved:line:bbcopy} takes constant time. 

By implementing this idea, one does not need to know the upper bound $\bbsol$ anymore: new simulations will be dynamically created as long as it is necessary to discover new solutions ahead. The fact that one has found every solution is still witnessed by the fact that $M[0]$ has moved without finding any solution in $Z_0$. This improvement has yet another advantage compared to \Cref{alg:enumit}: it has roughly the same total time as the original algorithm. Hence, if one is interested in generating every solution with a polynomial delay from $\abbox$, our method may make the maximal delay worse but does not change much the time needed to generate all solutions. 

\subparagraph{Correctness of  \Cref{alg:enumit_improved}.} Correctness of \Cref{alg:enumit_improved} can be proven in a similar way as for \Cref{alg:enumit}. \Cref{lem:c_vs_si} still holds for every state, where $N$ in the statement has to be replaced by $\length(M)-1$. The proof is exactly the same but we have to verify that when a new simulation is inserted into $M$, the property still holds. Indeed, let $c$ be a state that follows the insertion of a new simulation (Line~\ref{enumit_improved:line:newproc}). We have now $\length(M)-1 = N+1$ (thus the last index of $M$ is $N+1$). Moreover, we claim that $S_{N+1}^c=S_{N}^c$. Indeed, at this point, the simulation $M[N+1]$ has not output any solution. Moreover, by construction, $c_{N} = \bbstep[M[N]] = \bbsteps[M[N+1]] = c_{N+1}$. Since $c_{N} \geq 2d|S^c_{N}|$ by induction, we have that $c_{N+1} \geq 2d|S^c_{N+1}|$. Moreover, the following adaptation of \Cref{cor:enumitends} holds for \Cref{alg:enumit_improved}.

\begin{lemma}
   \label{lem:enumit_improved_ends} Let $c$ be the state reached when Algorithm~\ref{alg:enumit} stops. Then $N := \length(M)-1 = 1+\log(\sharp \abbox)$ and for every $i \leq N$, $c_i \geq 2^{i}d$. 
\end{lemma}
\begin{proof}
  The lower bound $c_i \geq 2^{i}d$ for $i \leq N$ is proven by induction exactly as in the proof of \Cref{lem:c_vs_si}. The induction holds as long as $2^{i-1} \leq \sharp \abbox$, because we need this assumption to prove that there are at least $2^{i-1}$ solutions in the interval $[1,2^{i-1}d]$.  Now, one can easily see that if $i \leq 1+\log(\sharp \abbox)$ and $c_i \geq 2^{i}d$  then the simulation $M[i]$ has reached $2^{i-1}d$ at some point and thus, has created a new simulation $M[i+1]$. Thus, by induction, the algorithms creates at least $1+\log(\sharp \abbox) = N$ new simulations. Thus $\length(M) \geq N+1$ (as $M$ starts with one element).

  Finally, observe that $M[N]$ outputs solutions in the zone $Z_N = [2^{N-1}d+1, 2^{N}d]$ and that $2^{N-1}d = \sharp \abbox$ which is an upper bound on the total time of $\abbox$. Thus, the simulation $M[N]$ will end without creating a new simulation. In other words, $\length(M)-1 = N$. 
\end{proof}

\subparagraph{Delay of Algorithm~\ref{alg:enumit_improved}.} 
% While establishing the correctness of Algorithm~\ref{alg:enumit_improved} is similar to the one of Algorithm~\ref{alg:enumit}, proving a bound on the delay of Algorithm~\ref{alg:enumit_improved} is not as straightforward. 
The bound on the delay is similar to the one of \Cref{alg:enumit}, but we have to take into account the additional $\bbstep[\cdot]$ instructions.
By Lemma~\ref{lem:enumit_improved_ends}, the size of $M$ remains bounded by $2+\log(\sharp \abbox)$ through the algorithm, so there are at most $2d(2+\log(\sharp \abbox))$ executions of $\bbstep[\cdot]$ between two solutions, for the same reasons as before. We have assumed that the execution of $\bbstep[\cdot]$ is in time and space complexity $O(1)$, thus the delay of \Cref{alg:enumit_improved} is $O(\log(\sharp\abbox)d)$. When we want to account for the complexity of the instructions, the use of $\bbstep[\cdot]$ makes it more complicated to avoid overhead, but it is still achievable by a lazy copy mechanism as explained in \Cref{sec:oracles}.

 % However, we also have to account for the execution of $\bbcopy[\cdot]$. When implemented naively, this operation requires a time $O(\bbspace)$ to copy the entire configuration of the simulation in some fresh part of the memory. It would add $O(\bbspace)$ to the delay of \Cref{alg:enumit_improved} compared to \Cref{alg:enumit}. However, one can amortize this $\bbcopy[\cdot]$ operation by lazily copying the memory while running the original simulation and by adapting the sizes of the zones so that we can still guarantee a delay of $O(\log(\sharp\abbox)d)$ in \Cref{alg:enumit_improved}. The method is formally described in \Cref{sec:copyorcales}.

\subparagraph{Total time of Algorithm~\ref{alg:enumit_improved}.} A minor modification of Algorithm~\ref{alg:enumit_improved} improves its efficiency in terms of total time. By definition, when simulation $M[i]$ exits $Z_j$, it does not output solutions anymore. Thus, it can be removed from the list of simulations. It does not change anything concerning the correctness of the algorithm. One just has to be careful to adapt the bounds in Algorithm~\ref{alg:enumit_improved}. Indeed, $2^jd$ is not the right bound anymore as removing elements from $M$ may shift the positions of the others. It can be easily circumvented by also maintaining a list $Z$ such that $Z[i]$ always contains the zone that $M[i]$ has to enumerate.

By doing it, it can be seen that each step of $\abbox$ having a position in $Z_i$ will only be visited by two simulations: the one responsible for enumerating $Z_i$ and the one responsible for enumerating $Z_{i+1}$. Indeed, the other simulations would either be removed before entering $Z_i$ or will be created after the last element of $M$ has entered $Z_{i+1}$. Thus, the $\bbstep[\cdot]$ operation is executed at most $2T$ times where $T$ is the total time taken by $\abbox$. Hence the total time of this modification of Algorithm~\ref{alg:enumit_improved} is $O(T)$. 

All previous comment on Algorithm~\ref{alg:enumit_improved} allows us to state the following improvement of Theorem~\ref{th:geometric}:
\begin{theorem}
    \label{th:geometric_improved} For every enumeration process $\abbox$, \Cref{alg:enumit_improved} with oracle $\abbox$ is an enumeration scheme, independent from $\bbsol$, with worst case delay $O(\bbam \cdot \log(\sharp \abbox))$, space $O(\log(\sharp \abbox))$ and total time $O(T)$ where $T$ is the total time of $\abbox$.
\end{theorem}

We sill need to know one parameter (or an upper bound) to run Algorithm~\ref{alg:enumit_improved}: $\bbam$, the amortized delay of the enumeration process.
When applying this regularization scheme to a RAM, taking into account the complexity of the blackbox operations, we also need to know a bound on the space complexity of the RAM to avoid any overhead, see \Cref{th:space_not_cardinal} which adapts \Cref{th:geometric_improved} to this context. 
% We now combine \Cref{alg:enumqueue_improved} and \Cref{alg:enumit_improved} to perform geometric regularization without knowing $\bbam$.

\subsection{Geometric Regularization with Unknown Delay} \label{sec:geometric_amortization_unknown_delay}

In this section, we propose an adaptation of geometric regularization,
which works \emph{without prior knowledge of the amortized delay}. It is based 
on the same method as \Cref{alg:enumqueue_improved}: maintaining a local bound on the amortized delay, and using a function of this bound as budget in \Cref{alg:enumit}. The proof is essentially the same as in \Cref{th:adaptative}.

\begin{algorithm}
    \KwInput{Oracle $\abbox$} % $x \in \Sigma^*$ of size $n$ and $K$ such that $K \geq \sharp I(x)$}
  \KwOutput{Enumerate $\abbox$ with delay bounded by $O(p\log(p)^2 \cdot \log(\bbsol))$ and space $O(\log(\bbsol) \cdot \bbspace)$ where $p=\amdelay[\abbox]$.}

  \Begin{
    $N \gets \lceil \log(\bbsol) \rceil$, $d \gets 8F$, $S \gets 0$\;
    \lFor{$i=0$  \KwTo $N$}{
      $M[i] \gets \bbload$
    }
    $j \gets N$\;
    \While{$j \geq 0$}{\label{line:loopwhile}
      \While{$B < d\log(d)^2$}{\label{line:loop2}
        $\bbstep[M[j]]$\;
        \If{$j=N$} {
          \lIf{$\bbisoutput[M[j]]$}{$S \gets S+1$}
          $d \gets \max(d, {\bbsteps[M[N]] \over S+1})$\;
        } 
        \If{$\bbisoutput[M[j]]$ \textbf{\emph{and}} $\bbsteps[M[j]] \in Z_j$ \label{line:unk:test}}{ 
          $\out(\bboutput[M[j]])$\;
          $j \gets N+1$\;
          \Break\;
        }
        $B \gets B+1$\;
        }
      $j \gets j-1$\;
      $B \gets 0$\;
      }
    }
  \caption{In the code, $Z_0=[0;2F[$ and $Z_j = [2^{j}F;2^{j+1}F[$ for $j>0$, where $F-1$ is the first time at which $\abbox$ outputs a solution. We assume $F$ is precomputed (for example, by running $M[N]$ until it finds a solution).}
  \label{alg:geometric-adaptative}
\end{algorithm}

\begin{theorem}
  \label{thm:geometric-adaptative} For every enumeration oracle $\abbox$ with amortized delay $p = \amdelay[\abbox]$, \Cref{alg:geometric-adaptative} with oracle $\abbox$ is an enumeration scheme independent from $\bbam$ with worst case delay $O(p\log(p)^2 \cdot \log(\bbsol))$ and space $O(\log(\bbsol))$.
\end{theorem}
\begin{proof}
  Recall from the caption of \Cref{alg:geometric-adaptative} that $F-1$ is the time where $\abbox$ outputs its first solution and $Z_0 = [1,F[$, $Z_i = [2^{i}F, 2^{i+1}F[$. We let $d_i$ be the amortized delay of $\abbox$ at the beginning of $Z_i$, that is, $d_i = \max_{t \leq 2^{i}F} {t \over 1+k_t}$, where $k_t$ is the number of solutions output by $\abbox$ before time $t$. In particular, $d_i \geq {2^{i}F \over (1+S_i)}$ where $S_i$ is the number of solutions output by $\abbox$ before the zone $Z_i$, that is, before time $2^{i}F$. 

  First of all, observe that $d$ is always bounded by $p \log(p)^2$ in \Cref{alg:geometric-adaptative}, for the same reasons as in the proof of \Cref{th:adaptative}. Hence, as in the proof of \Cref{th:geometric}, the delay between two outputs is at most $O(d\log(d)^2 \cdot N) = O(p\log(p)^2 \cdot \log(\bbsol))$. 
  
  It remains to prove that \Cref{alg:geometric-adaptative} outputs every solution of $\abbox$. We prove it by contradiction. Assume that when the algorithm stops, there is a zone $Z_l$ that has not been fully visited by $M[l]$ (that is $\bbsteps[M[l]] < 2^{l+1} F-1$ when the algorithm stops). Without loss of generality, we assume $Z_l$ is the first such zone, that is, $\bbsteps[M[i]] \geq 2^{i+1}F-1$ for every $i \leq l-1$.

  First, observe that $l > 0$ since the algorithm stops when $M[0]$ does not visit any solution in $Z_0$ after $d(\log d)^2$ steps. But since $d>F$, $M[0]$ is necesseraly out of $Z_0$ when the algorithm stops.

  Let $P_l$ be the number of solutions that have been output by $M[0],\dots,M[l-1]$ before the algorithm stops. By definition of $l$, $P_l \geq S_l$ since every solution in $Z_i$ have been enumerated by $M[i]$ for $i < l$ and there are by definition $S_l$ such solutions. We now prove that $P_l < S_l$ which would hence result in a contradiction.

  As before, $M[l]$ is always ahead of $M[l-1]$. Hence, $M[l]$ is necesserily out of $Z_{l-1}$ since $M[l-1]$ is out of $Z_{l-1}$ by definition of $l$. Now, in \Cref{alg:geometric-adaptative}, $M[l]$ moves each time a solution is found by $M[i]$ in $Z_i$ for $i < l$. During the execution of the algorithm, whenever $M[l]$ is in $Z_i$ and a solution is found by $M[j]$ with $j < l$, then $M[l]$ has just moved by at least $d_i (\log d_i)^2$. Indeed, at this point in the algorithm,  $d \geq d_i$ because $M[N]$, which is ahead of $M[l]$, has updated $d$ when entering $Z_i$ (recall $d_i$ is the approximate delay at the beginning of $Z_i$). Hence,  at most $|Z_i| \over d_i (\log d_i)^2$ solutions have been output while $M[l]$ is in $Z_i$. Hence, $P_l$, the number of solutions output by $M[0], \dots, M[l-1]$ before the algorithm stops is:
  \[ P_l \leq \sum_{i=0}^{l-1}{|Z_i| \over d_i (\log d_i)^2}.\]
  This equality is very similar to the one in \Cref{th:adaptative}, where it was shown that $P_l \leq S_l$. The only difference here is that $|Z_i| = 2^iF$ instead of $2^i$ but the same proof still works. We bound the first terms as before from the fact that $d_i \geq 8$:

     \begin{align*}
     \sum_{i=0}^{l-\lfloor \log(d_l) \rfloor + 1} \frac{|Z_i|} {d_i\log(d_i)^2}
     & \leq \sum_{i=0}^{l-\lfloor \log(d_l) \rfloor + 1} \frac{|Z_i|}{64} \\ 
     & \leq \frac{F}{64} \sum_{i=0}^{l- \lfloor \log(d_l) \rfloor + 1} 2^{i} , \text{ since } |Z_i|=2^iF \\
     & \leq \frac{F}{64} 2^{l-\lfloor \log(d_l) \rfloor +2} \leq \frac{2^{l}F}{8d_l} \\
     & \leq \frac{S_l+1}{8}.
   \end{align*}
   As before, the last inequality follows from the fact that $d_l$ is defined to be $\max_{t \leq 2^lF} {t \over k_t+1}$ where $k_t$ is the number of solutions output at time $t$. We can upperbound $t$ by $2^lF$ and lowerbound $k_t$ by $S_l$.

   To bound the last terms of the sum, a similar reasoning as in \Cref{th:adaptative} allows us to prove that $d_{i+1} \leq  2d_{i}$, hence, $d_{l-i} \leq 2^{-i}d_{l}$ for every $i$.

   We thus have:
      \begin{align*}
     \label{eq:adaptative-large}
    \sum_{i=l-\lfloor \log(d_l) \rfloor +2}^{l-1} {|Z_{i}| \over d_{i}\log(d_{i})^2} & <  \sum_{j=0}^{\lfloor \log(d_l) \rfloor -2} {|Z_{l-j}| \over d_{l-j}\log(d_{l-j})^2} \\
     & \leq \sum_{j=0}^{\lfloor \log(d_l) \rfloor-2} {2^{l-j}F \over 2^{-j}d_l\log(2^{-j}d_l)^2} \text{ from what preceeds} \\
     & \leq {2^{l}F \over d_l} \sum_{j=0}^{\lfloor \log(d_l) \rfloor-2} {1 \over (\log(d_l) - j)^2} \\
                                                                                   & \leq {S_l+1} ({\pi^2 \over 6}-1) \text{ as in \Cref{th:adaptative}}\\
    &  \leq {3\over 4} (S_l+1) 
  \end{align*}
Hence $P_l < {7 \over 8}(S_l+1)$, that is $P_l < S_l$ since $S_l>0$ is an integer. This is a contradiction. Hence, every zone $Z_l$ have been visited by $M[l]$ when the algorithm stops. That is, every solution of $\abbox$ have been enumerated. 
\end{proof}

\subsection{Geometric Regularization for incremental delay}\label{sec:Inci}

A phenomenon that naturally appears in enumeration algorithms is that the first solutions are often easier to find than the last one. For example, saturation algorithms generate solutions by applying some polynomial time rules to enrich the set of solutions until saturation. There are many saturation algorithms, for instance, to enumerate circuits of matroids~\cite{khachiyan2005complexity} or to compute closure by set operations~\cite{mary2016efficient}. In this case, the larger the set is, the longer it is to apply the rule to every solution and check whether a new one appears. 

This phenomenon can be formalized as having an \emph{incremental delay}, that is, a delay that not only depends on the input size but also on the number of solutions output so far. Recall that for a machine $M$, we denote by $T_M(x,i)$ the time when $M$ outputs its $i^\th$ solution. For $f \colon \N \rightarrow \N$, the \emph{$f$-incremental delay of $M$ on input $x$} is defined as $\max_i {T_M(x,i+1)-T_M(x,i) \over f(i)}$. In other words, if $M$ has $f$-incremental delay $d(x)$, then the time between the output of the $i^\th$ solution and of the $(i+1)^\th$ solution is at most $d(x)f(i)$. The delay of $M$ is exactly its $f$-incremental delay for $f(i)=1$.

As before, this notion corresponds to a worst-case notion, where we take into account the longest time one has to wait between the output of two solutions. We introduce an amortized version of this definition as follows: for $f \colon \N \rightarrow \N$, the \emph{amortized $f$-incremental delay of $M$ on input $x$} is defined as $\max_i {T_M(x,i) \over f(i)}$. In other words, if $M$ has amortized $f$-incremental delay $d(x)$, then the time to enumerate $i$ solutions is bounded by $d(x)f(i)$. The amortized delay of $M$ is exactly its amortized $f$-incremental delay for $f(i)=i$.

We denote by $\WInc{a}$ the set of enumeration problems in $\EnumP$ that can be solved by a machine $M$ whose $i^a$-incremental delay is polynomial in $|x|$\footnote{This corresponds to the classe $\UInc{a}$ of the conference version of this paper where the notions where not unified in the same way as it is now. The WC stands for \emph{worst case}.}. Observe that $\WInc{0} = \DelayP$.  Similarly, we denote by $\AmInc{a}$ the set of enumeration problems  $\EnumP$  that can be solved by a machine $M$ whose amortized $i^a$-incremental delay is polynomial in $|x|$\footnote{This corresponds to the classe $\Inc{a}$ of the conference version of this paper.}. Again, observe that $\AmInc{1} = \AmDelayP$. 

The class $\bigcup_{a\geq 1} \AmInc{a}$ is believed to be strictly included in $\OutputP$, the class of problems solvable in total polynomial time, since this is equivalent to $\TFNP \neq \FP $~\cite{capelli2019incremental}. Moreover, the classes $\AmInc{a}$ form a strict hierarchy assuming the exponential time hypothesis~\cite{capelli2019incremental}.

\begin{algorithm}
  \KwInput{Oracle $\abbox$} % $x \in \Sigma^*$ of size $n$ and $K$ such that $K \geq \sharp I(x)$}
  \KwOutput{Enumerate $\abbox$ with incremental delay bounded by $O(i^a d \cdot \log(\bbsol))$ and space $O(\log(\bbsol))$}
  \Begin{
    $N \gets \lceil \log(\bbsol) \rceil$, $S \gets 1$\;
    \lFor{$i=0$  \KwTo $N$}{
      $M[i] \gets \bbload$
    }
    $j \gets N$\;
    \While{$j \geq 0$}{\label{enumit:line:while}
      \For{$B \gets 2d \cdot S^a (a+1)$  \KwTo $0$}{\label{enumit:line:for}
        $\ramstep(M[j])$\;
        \If{$\bbisoutput[M[j]]$ \textbf{\emph{and}} $\bbsteps[M[j]]
          \in Z_j$ \label{line:incp:test}}{ 
          $\out(\bboutput[M[j]])$\;
          $S\gets S+1$\;
          $j \gets N$\;
          \Break;
        }
      }
      \lIf{$B = 0$}{$j \gets j-1$}
    }
  }
  \caption{Geometric regularization for incremental time. In the code, $d$ is an upper bound on the amortized $i^{a+1}$-incremental delay of $\abbox$  and $Z_0=[0;d[$ and $Z_j = [2^{j-1}d+1;2^jd[$ for $j>0$.}
  \label{alg:enumitinc}
\end{algorithm}

As before, incremental delay and amortized incremental delay are closely related and one can use the idea from \Cref{alg:enumit} to transform any algorithm with amortized $i^{a+1}$-incremental delay into an algorithm with $i^a$-incremental delay. Pseudo code is given in \Cref{alg:enumitinc}. The following holds:

\begin{theorem}
  \label{thm:geometricinc-guarantees} If $d$ is an upper bound on the amortized $i^{a+1}$-incremental delay of $\abbox$, then \Cref{alg:enumitinc} is a regularization scheme with $i^a$-incremental delay $O(d \cdot \log(\bbsol))$ and space $O(\log(\bbsol))$.
\end{theorem}
\begin{proof}
\Cref{alg:enumitinc} is similar to \Cref{alg:enumit} but we maintain a counter $S$ of the number of output solutions and modify the initialization of $B$ in the for loop at line~\ref{line:loop2} to $2d \cdot (a+1)S^{a}$.  

Obviously, the delay between two solutions in \Cref{alg:enumitinc} is bounded by $2d \cdot S^a(a+1)\log(\bbsol)$ and $S$ is the number of solutions output up to this point in the algorithm. Hence the $i^a$-incremental delay of \Cref{alg:enumitinc} is $2d(a+1)\log(\bbsol)=O(d \cdot \log(\bbsol))$. 

It remains to prove that all solutions are enumerated by the algorithm. Assume that the first $i+1$ machines $M[0],\dots,M[i]$ have output all the solutions in their zones, then we prove as in Corollary~\ref{cor:enumitends}, that the machine $M[i+1]$ has also output all its solutions. The number of solutions output by $M[0],\dots,M[i]$ is the number of solutions output by $\abbox$ up to time step $2^id$. Let $s_i$ be this number, then $s_i^{a+1}d \geq 2^id$ since $d$ is an upper bound on the $i^{a+1}$-incremental time of $\abbox$.  That is, $s_i \geq 2^{i/(a+1)}$. 

When a solution is output by a machine $M[j]$ with $j \leq i$, then $j$ is set to $N$ and all machines $M[k]$ with $k>i$ move by at least $2d \cdot S^{a}(a+1)$ steps where $S$ is the current number of output solutions before $M[i]$ moves again.
Hence, we can lower bound the number of moves of the machine $M[i+1]$ by $\sum_{S=0}^{s_i} 2d \cdot S^a(a+1) \geq 2d(a+1) \cdot \sum_{S=0}^{2^{i/(a+1)}} S^a$. 
Since $\sum_{S=0}^{n} S^a \geq \int_{0}^{n} S^a \,dS \geq n^{a+1}/(a+1)$, the number of moves of $M[i+1]$ is larger than $2^{i+1}d$ which is the upper bound of its zone.
\end{proof}

 For $a \geq 0$, we denote by $\WInc{a}^{\poly}$ (respectively $\AmInc{a}^{\poly}$), the class of problems from $\WInc{a}$ (respectively $\AmInc{a}$) that can be solved in polynomial space. \Cref{thm:geometricinc-guarantees} allows to prove the following generalisation of \Cref{thm:amdelayp_is_amdelayp_poly} to incremental delay:

\begin{theorem}
   \label{th:uincvsinc}
 For all $a \geq 0$, $\AmInc{a+1}^{\poly} = \WInc{a}^{\poly}$.
 \end{theorem}

 \begin{proof}
The inclusion $\WInc{a}^{\poly} \subseteq \AmInc{a+1}^{\poly}$ is straightforward and follows by a simple computation of the time to generate $i$ solutions, see \cite{capelli2019incremental}. Indeed, if $d$ is the $i^a$-incremental delay of $M$ then we have $T_M(x,i) - T_M(x,0)= \sum_{j=0}^{i-1} T_M(x,j+1)-T_M(x,j) \leq \sum_{j=0}j^ad(x) \leq j^{a+1}d(x)$. 

The inclusion $\AmInc{a+1}^{\poly}  \subseteq  \WInc{a}^{\poly}$ follows by \Cref{thm:geometricinc-guarantees}. Indeed, let $M$ be a machine solving a problem $P$ in $\AmInc{a+1}^\poly$ such that there is a polynomial $p$ where the amortized $i^{a+1}$-incremental delay of $M$ on input $x$ is at most $p(|x|)$ and a polynomial $s$ such that the space used by $M$ on input $x$ is at most $s(|x|)$. Since $P$ is in $\EnumP$, we also know that there is a polynomial $q$ such that the solutions of $P$ are of size at most $q(|x|)$. Given $x$, we call \Cref{alg:enumitinc} with blackbox $M$ on input $x$ and we let $\bbsol = 2^{q(|x|)}$ and $d = p(|x|)$ which are upper bounds on the number of solutions and on the $i^{a+1}$-incremental delay of $M$ respectively. By \Cref{thm:geometricinc-guarantees}, we get an algorithm solving $P$ with $i^a$-incremental delay $O(p(|x|) q(|x|))$ and space $O(s(|x|)q(|x|))$ which are both polynomial in $|x|$. That is, $P$ is a problem in $\WInc{a}^{\poly}$.
 \end{proof}
%%% Local Variables:
%%% mode: latex
%%% TeX-master: "theoretics"
%%% End:

\section{Lower bounds}
\label{sec:lower-bounds}
 
In this section, we prove limits of regularization schemes for enumeration algorithms in the blackbox model. In particular, we show that when the amortized delay of the blackbox is not known, one cannot obtain a worst case delay that is linear in the amortized delay of the blackbox. Later, we show that any regularization scheme that preserves the order of the solutions in the output needs either to have a delay or space that is exponential in the size of the solutions, showing in particular that the guarantees offered by geometric amortization cannot be met if one wants to preserve order.

As before, we assume that each blackbox operation is executed in constant time. Observe that in case of lower bounds, this assumption only makes our results stronger. Our lower bounds are actually only based on the number of $\bbstep$ executed, hence we do not even take into account the rest of the computation done by the regularization scheme. When evaluating the space used, we consider that each $M$ given by $M \gets \bbload[\abbox]$ is stored in constant space, as explained in \Cref{sec:bufferamortization}.

\subsection{Regularization delay}

Before focusing on lower bounds, we must clarify the nature of the results presented in this section. We are interested in proving lower bounds on the worst case delay of regularization schemes. By definition, when a regularization scheme $\calM$ is given access to a blackbox enumeration process $\abbox$, it enumerates the same solutions as $\abbox$. The worst case delay of $\calM$ can hence be seen as a function mapping a particular blackbox $\abbox$ to an integer $d_\abbox \in \N$, being the maximal time between two outputs of $\calM$ when given access to $\abbox$. This definition makes it hard to compare the behavior of $\calM$ on families of blackboxes. Indeed, when defined this way, it is not clear what parameters of $\abbox$ are acting on the delay of $\calM$.

%To clarify this situation, we say that two blackboxes $\abbox$ and $\abbox'$ are equivalent if the following hold: $\bbam = \bbam[\abbox'], \bbsol=\bbsol[\abbox'], \bbsize=\bbsize[\abbox']$ and if at every time $t$, either both $\abbox$ and $\abbox'$ output the same solution $s$ or they both do not output any solution.  Observe that the behavior of a regularization scheme $\calM$ will be the same on two equivalent blackboxes

Now given $p$ and $b$, let $I(p,b)$ be the set of blackboxes with amortized delay at most $p$ and solution size at most $b$. Observe that $I(p,b)$ is finite. Indeed, a blackbox in $I(p,b)$ can output at most $2^b$ distinct solutions in at most $2^bp$ time steps and its values $\bbsize, \bbsol$ and $\bbam$ are bounded by $b$, $2^b$ and $p$ respectively. Hence, we can define without ambiguity the \emph{regularization delay of a regularization scheme $\calM$} as a function of $p$ and $b$ as follows:  a regularization scheme $\calM$ has regularization delay $r(a,b)$ if for every $a,b$, $r(a,b) \geq \max_{\abbox \in I(a,b)} d_\abbox$ where $d_\abbox$ is the worst case delay of $\calM$ when given access to $\abbox$. The regularization delay increases with $p$ and $b$ since $I(p,b) \subseteq I(p',b')$ if $p \leq p'$ and $b \leq b'$. 

That way, we know that if a regularization scheme $\calM$ has regularization delay $r(p,b)$ and is given access to a blackbox with amortized delay at most $p$ and solution size at most $b$, then $\calM$ outputs every solution of $\abbox$ with worst case delay at most $r(p,b)$. In this section, we prove some limits on the achievable regularization delay.

\subsection{Unknown Delay}

\label{sec:lower_bound_unknown_delay}

\Cref{alg:enumit_improved} has been shown in \Cref{th:adaptative} to be a regularization scheme independent from $\bbam$ with worst case delay $O(bp (\log p)^2)$, that is, it automatically adjusts to the unknown amortized delay $p$ of the blackbox $\abbox$. We have later shown that this worst case delay can be brought down to $O(bp(\log p)^{1+\varepsilon})$ for any $\varepsilon$. A natural question is hence to understand whether we can bring down the worst case delay of such regularization scheme further. We show in this section that obtaining a linear dependency in $p$ is impossible.
More precisely, we prove the following:
\begin{theorem}\label{thm:lower_bound_udelay}
  There is no regularization scheme with regularization delay $O(p \cdot \poly(b))$.
\end{theorem}

The proof of \Cref{thm:lower_bound_udelay} mainly rely on the following lemma which exhibits a family of adversarial blackboxes:

\begin{lemma}\label{lem:lower_bound_udelay}
  Let $\calM$ be a regularization scheme. For every $s \in \N$, there exists a blackbox enumeration process $\abbox_s$ that outputs the integers from $0$ to $s$ with amortized delay $p_s := \amdelay[\abbox_s]$ and such that the worst case delay of $\calM$ on $\abbox_s$ is at least $s \cdot p_s$.
\end{lemma}
\begin{proof}
To lower bound the delay of $\calM$, we only evaluate the number of $\bbstep[\cdot]$ instructions it executes. We let $\adv[][s][t]$, with $t>s$, be an enumeration process which outputs the integers from $0$ to $s-1$ in the first $s$ time steps and, when it has executed $t$ steps, outputs $s$ and stops. We let $\bbsize[\adv[][s][t]] = \lceil \log(s+1) \rceil$, $\bbsol[\adv[][s][t]] = s+1$ and $\bbam[\adv[][s][t]] = +\infty$.
 
Let $F_{s,t}(i)$ be the time step at which $\calM$ outputs the solution $i$ when given access to blackbox $\adv[][s][t]$. Assume that for some $t$, $F_{s,t}(i) < t$,
that is $\calM$ outputs $i$ before having simulated $\adv[][s][t]$ to the end. Before time $t$, $\calM$ cannot differentiate between the processes $\adv[][s][t']$ for $t' \geq t$ because the results of the first $t$ $\bbstep[\cdot]$ instructions are the same for all these processes and $\calM$ cannot use $\bbam[\cdot], \bbsize[\cdot], \bbsol[\cdot]$ to distinguish them. Hence, for all $t' \geq t$, $F_{s,t'}(i) =  F_{s,t}(i)$.  When there is a $t$ such that $F_{s,t}(i) < t$, we denote by $F_s(i)$ the value $F_{s,t}(i)$ otherwise  $F_s(i)= +\infty$.

We first assume that there exists at least one $j$ such that $F_s(j) < +\infty$ and we let $i$ be the largest $j$ such that $F_s(j) < +\infty$. Observe that it exists since $F_s$ is defined on integer smaller than $s$.
We consider a run of $\calM$ with blackbox access $\adv[][s][F_s(i) + T]$ where $T$ is a value that we will fix later.
For  $s+1 \geq j > i$, we have $F_s(j) = +\infty$ since $i$ is maximal. In particular, $\calM$ with blackbox $\adv[][s][F_s(i) + T]$ enumerates $(i+1)$ after the time $F_s(i)+T$ and, by definition, $i$ is output at time $F_s(i)$.  
Hence, the delay of $\calM$ when given $\adv[][s][F_s(i) + T]$ as blackbox is at least $T$. The amortized delay of $\adv[][s][F_s(i) + T]$ is maximal when generating the last solution, since the first $s$ solutions are produced during the first $s$ time steps. Hence, the amortized delay of $\adv[][s][F_s(i) + T]$ is exactly $(F_s(i) + T)/(s+1)$.
 
Using the bound on the delay of $\calM$ and the amortized delay of $\mathcal{A}(s, F_s(i) + T)$, their quotient is at least $$r_T = \frac{(s+1)T}{(F_s(i) + T)}.$$ 
The limit of $r_T$ when $T$ goes to infinity is $s+1$, since $F_s(i)$ does not depends on $T$. Therefore, there is a $T_0$ such that $r_{T_0} > s$. We let $\abbox_s$ be $\adv[][s][F_s(i)+T_0]$ and $p_s$ be its amortized delay. By definition of $T_0$, the delay of $\calM$ given blackbox access to $\abbox_s$ is at least $s \cdot p_s$. Since $\abbox_s$ outputs $s+1$ solutions, it matches the conditions of the lemma's statement. % $\sharp \adv[][s][F_s(i) + T_0]= s + 1$, hence the delay of $\calM$ is at least $\sharp \adv[][s][F_s(i) + T_0]-1$ times the incremental delay of $\adv[][s][F_s(i) + T_0]$.

Now, consider the remaining case, that is for all $i$ and all $t$, $F_{s,t}(i) \geq t$. 
Then, the delay before the enumeration of the first solution of $I_{s,t}$ by $A$ is at least $t$, thus its delay is at least $t$. The amortized delay of $\adv[][s][t]$ is equal to $t/(s+1)$ as explained in the first part of the proof, therefore we can define $\abbox_s = \adv[][s][s]$ and $p_s = 1$ its amortized delay. We indeed have that the worst case delay of $\calM$ on $\abbox_s$ is at least $s$, that is, it is greater than $s \cdot p_s$.
\end{proof}

\begin{proof}[Proof of \Cref{thm:lower_bound_udelay}]
  Assume toward a contradiction that there exists a regularization scheme $\calM$, a constant $C$ and a polynomial $q$ such that the regularization delay of $\calM$ is $C \cdot q(b) \cdot p$. Let $s_0$ be such that $C q(\lceil \log s_0+1 \rceil ) < s_0$. However, the worst case delay of $\calM$ on $\abbox_{s_0}$ from \Cref{lem:lower_bound_udelay} is at least $s_0 \cdot \amdelay[\abbox_{s_0}] > C q(\lceil \log s_0+1 \rceil ) \cdot \amdelay[\abbox_{s_0}]$. But the solution size of $\abbox_{s_0}$ is $\lceil \log s_0+1 \rceil$, hence the worst case delay of $\calM$ on blackbox $\abbox_{s_0}$ exceeds the assumed upper bound, contradiction.
\end{proof}

The proof of \Cref{lem:lower_bound_udelay} uses the enumeration processes $\adv[][s][t]$ given as a blackbox. If we want to strengthen the theorem and ask for the blackbox of the regularization scheme to be obtained from a RAM, it cannot be done for $\adv[][s][t]$. Indeed, one would need time to generate distinct elements and to maintain a counter to know when to stop. So it cannot output $s$ distinct solutions in $s$ steps exactly. However, it is easy to see that this is possible to output $s$ distinct solutions in time $\alpha s$ for some constant $\alpha$. For example, $\adv[][s][t]$ could output a Gray Code encoding of $0,\dots,s$ (see Section~\ref{sec:counters} for details). The proof of \Cref{lem:lower_bound_udelay} can easily be adapted to these existing RAM by considering $\mathcal{A'}_{s,t}$ which outputs $0,\dots,s-1$ in $\alpha s$ steps and $s$ at step $\alpha t$. The same analysis would prove the existence of similar adversarial $\abbox_s$.
  % We decided not to include these details in the proof of \Cref{lem:lower_bound_udelay} to focus on the core technique.

 We conclude this section by observing that there is still a gap between our upper and lower bounds. Namely, \Cref{thm:adaptative-eps} shows that there exists regularization schemes with regularization delay $O((b+p)(\log (b+p))^{1+\varepsilon})$ for any $\varepsilon>0$ and \Cref{thm:lower_bound_udelay} shows that there exist no regularization scheme with worst case delay $O(\poly(b) \cdot p)$. Observe that we do not even need to assume the scheme to be independent from $\bbam$ for the lower bound: even if the regularization scheme uses $\bbam$ as in \Cref{alg:enumit}, the regularization delay will not be linear in $p$. It can be linear in $\bbam$, which is different since we have no guarantee that the given upper bound $\bbam$ is good.  We leave the following question open:

 \begin{openproblem}
   Is there a regularization scheme in the blackbox model with regularization delay $O(\poly(b) \cdot p\log p)$?
 \end{openproblem}
 
%Observe that the bound we give in Theorem~\ref{prop:lower_bound} is tight, since we can always enumerates the solutions of some enumeration algorithm $\abbox$ in total time $\sharp \abbox p$, where $p$ is the amortized delay of $\abbox$, since it is a bound on the delay of $\abbox$ itself. 
%With knowledge of the incremental delay $p(n)$ of an enumeration algorithm $\abbox$, we obtain by Proposition~\ref{prop:queue} the same delay for the amortization up to a constant factor. When the incremental delay is unknown, and $\abbox$ and its input are given as an oracle, it is not possible to obtain an amortization which is in $O(p(n))$. More precisely, there is no algorithm $\calM$ and a constant $C \in \mathbb{N}$, such that for all enumeration algorithms $\abbox$ with incremental delay $p(n)$, enumerates all solutions of $\abbox$ in worst case delay at most $Cp(n)$. It is a consequence of Theorem~\ref{prop:lower_bound}, where $\sharp \abbox(x)$ is chosen to be larger than $C$ and $p(n) = p$ is given by the theorem. 

%In Theorem~\ref{prop:lower_bound}, the machine $\calM$ does no preprocessing. However, if the amortized algorithm $\abbox$ solves a problem in $\EnumP$, then $s$ the number of solutions can be exponential in the size of the input. If we allow a polynomial time preprocessing in the size of the input for $\calM$, its contribution to the delay is negligible when $s$ grows and we obtain a theorem similar to Theorem~\ref{prop:lower_bound}.  \flo{mmmh, this is weirdly phrased. Maybe we should rephrase \Cref{prop:lower_bound} so that it says delay/preproc?}

\subsection{Ordered Enumeration}
\label{sec:ordered_enum}

One drawback of geometric regularization is that the order of the output is modified. The buffer based regularization scheme presented in \Cref{alg:enumqueue} preserves the order but uses extra space, possibly linear in the number of output solutions. In this section, we provide a lower bound suggesting that there is an unavoidable tradeoff between the space and the worst case delay for any regularization scheme that preserves the order.

Recall that a regularization scheme $\calM$ is said to be \emph{order preserving} if for every enumeration process $\abbox$ it is given as oracle, $\calM$ enumerates the same set as $\abbox$ in the same order. We show:

\begin{theorem}
  \label{thm:lower-bound-ordered} There is no order preserving regularization scheme with regularization delay $d(p,b)$ and using space $s(p,b)$ such that $d$ and $s$ are both polynomial in $b$.%  and ? \bbam$, $\bbsize$ 
\end{theorem}

\Cref{thm:lower-bound-ordered} is a direct consequence of the following lemma:
\begin{lemma}\label{lem:lower-bound-order}
Let $\calM$ be an order preserving regularization scheme with worst case delay $d(p,b)$ and using space $s(p,b)$. For every $N \in \N$, there exists a blackbox $\abbox$ such that $\bbsol = N$, $b := \bbsize = 2\log(N)+1$ and $p := \bbam = 2$ such that $N \leq 9 d(p,b) \cdot s(p,b)^2$.
%  There is no RAM with black box access to enumeration algorithm of amortized delay $2$ outputting $N$ solutions, that enumerates the solutions in the same order with delay $D$ and space $S$, when $9DS^2 \leq N$.
\end{lemma}
\begin{proof}
  Assume toward a contradiction that there exists $N \in \N$ such that for any blackbox $\abbox$ with $\bbsol=N$, $\bbsize=2\log(N)+1$ and $\bbam=2$, we have $N > 9d(p,b)s(p,b)^2$. To ease notation, we let $D := d(p,b)$ and $S := s(p,b)$. In particular, $N > 9DS^2$. 

  We construct a family of blackboxes $(\ladv)_{\sigma}$ parametrized by some list $\sigma$ with $\bbam[\ladv]=2$, $\bbsol[\ladv]=N$ and $\bbsize[\ladv] = 2\log(N)+1$ and show that $\calM$ fails to correctly enumerate the solutions of $\ladv$ for every $\sigma$, which would result in a contradiction. 
  
W let $T = 18SD^2 - 6S$ and we consider the following family of adversarial blackboxes: given $\sigma=(s_1, \dots, s_{3S})$ a list $3S$ integers in $\{1,\dots,T\}$, the blackbox $\ladv$ first enumerates the binary representation of integers from $1$ to $T_0 := 9DS^2-6S$ with delay one. Then, it enumerates in order the elements of $\sigma$, concatenated with the binary representation of integers from $T_0+1$ to $T_0+3S$, with delay $3DS$. More precisely, it enumerates $s_i \odot [T_0+i]_2$ for $i = 1, \dots, 3S$ where $[j]_2$ is the binary representation of integer $j$ and $\odot$ denotes concatenation. Observe that this ensure that every integer output by $\ladv$ are distinct.

Now each element of $\sigma$ can be represented with $\log(T) \leq 1+\log(N)$ bits since we assumed $N > 9DS^2$, that is, $2N > T$. Similarly, since $T_0+3S < T/2$, we can encode any integer for $T_0+1$ to $T_0+3S$ with $\log(N)$ bits. Hence, every solution of $\ladv$ can be encoded with at most $2\log(N)+1$ bits a,d we can set $\bbsize[\ladv] = 2\log(N)+1$ for the blackbox. 

Now observe that for any $\sigma$, $\ladv$ outputs $9DS^2-3S$ distinct solutions, which is less than $N$ by assumption. Hence we can set $\bbsol[\ladv] = N$.

Finally, the amortized delay of $\ladv$ is $2$. Indeed, during the first $T_0$ steps, the algorithm has amortized delay $1$. Then it outputs the $(T_0+k)^\th$ solution at time $T_0+3kDS$. Hence the amortized delay at this point is at most $T_0+3kDS \over T_0+k$ and $k \leq 3S$. This ratio grows with $k$ and hence is maximal when $k=3S$. In this case, ${T_0+9DS^2 \over T_0+3S} = {18DS^2-6S \over 9DS^2-3S} = 2$. Hence the amortized delay of $\ladv$ is $2$ and we can set $\bbam[\ladv]=2$.

From what precedes, we can conclude that for every $\sigma$, $\calM$ with blackbox access to $\ladv$ and parameters $\bbam[\ladv]=2, \bbsize[\ladv]=2\log(N)+1$ and $\bbsol[\ladv]=N$ enumerates the solutions of $\ladv$, in the same order, with delay at most $D$ and space at most $S$ (remember $D$ and $S$ have been chosen to be maximal over every blackbox with such parameters). Observe that $T$ is the total time used by $\ladv$, which does not depend on $\sigma$.

We denote by $\Sigma=\{\sigma \mid \sigma=(\sigma_1,\dots,\sigma_{3S}) \text{ and } 1 \leq \sigma_i \leq T \}$. It is clear that $|\Sigma|=T^{3S}$. For $\sigma \in \Sigma$, we let $\calM_\sigma$ be the process obtained by running $\calM$ with $\ladv$ as oracle and we let $t_{\sigma}$ be the time at which $\calM_\sigma$ outputs the solution $T_0-1$, that is, the one right before the first solution that depends on $\sigma$.

Let $\sigma \in \Sigma$. Observe that since $\calM_\sigma$ uses space at most $S$ by definition, it does not use registers with index greater than $S$ in the memory. We define the \emph{footprint of $\sigma$}, denoted by $f_\sigma= (L_\sigma, K_\sigma, F_\sigma)$ as:
\begin{itemize}
\item $L_\sigma$ a list of size $S$ where for every $i \leq S$, $L[i]$ is the binary representation of the value of register $R_i$ at time $t_\sigma$ if $R_i$ contains a value. Otherwise, if $R_i$ contains a pointer toward a simulation of $\ladv$, $L[i] = \star$. 
\item $K_\sigma$ is a list of size $S$ such that if $R_i$ contains a pointer toward a simulation of $\ladv$ at time $t_\sigma$ that has been moved by $t$ steps, then $K_\sigma[i] = t$. Otherwise, if $R_i$ contains a value, $K_\sigma[i]=\bot$.
\item $F_\sigma$ is a list of size $S$ such that if $R_i$ contains a pointer toward a simulation of $\ladv$ at time $t_i$, then $F_\sigma[i]$ is a list of size $3DS$ such that entry $j$ of $F_\sigma[i][j]$ is $e$ if $\ladv$ outputs solution $e$ at time $t_i+j$ and $\bot$ otherwise. If $R_i$ contains a value, $F_\sigma[i] = []$. 
\end{itemize}

We say that $\sigma,\sigma' \in \Sigma$ are indistinguishable if they have the same footprint. Intuitively, it means that when $\calM_\sigma$ is about to output the first value of $\sigma$, its memory state is the same as when $\calM_{\sigma'}$ is about to output $\sigma'$. Moreover, the simulation of $\ladv$ that they keep in memory behave similarly for the next $3DS$ steps. 

We will show later by a counting argument that there exists $\sigma \neq \sigma'$ such that $\sigma$ and $\sigma'$ are indistinguishable. Before that, let us explain why this would be a contradiction, invalidating our initial hypothesis that $N > 9DS^2$. It is easy to see that if $\sigma$ and $\sigma'$ are indistinguishable, then $\calM_\sigma$ between time $t_\sigma$ and time $t_\sigma+3DS$ will behave exactly in the same way as $\calM_{\sigma'}$ between time $t_{\sigma'}$ and $t_{\sigma'}+3DS$. Indeed, the memory state of $\calM_\sigma$ at time $t_\sigma$ is exactly the same as the memory state of $\calM_{\sigma'}$ at time $t_{\sigma'}$ by definition. Moreover, every blackbox operation on $\ladv$ done by $\calM_\sigma$ between $t_\sigma$ and $t_\sigma+3DS$ will return exactly the same thing than $\calM_{\sigma'}$ between $t_{\sigma'}$ and $t_{\sigma'}+3DS$ since $F_\sigma=F_{\sigma'}$. Hence the values output by $\calM_\sigma$ and by $\calM_{\sigma'}$ are the same on this time window. However, by hypothesis, $\calM_\sigma$ and $\calM_{\sigma'}$ have delay at most $D$ and it remains $3S$ solutions to be output. Hence it means that $\calM_\sigma$ (respectively $\calM_{\sigma'}$) outputs the same values before finishing their computation, which is impossible since $\sigma \neq \sigma'$. 

It remains to show that there exists $\sigma \neq \sigma'$ that are indistinguishable. We prove it by a counting argument, showing that the number of distinct footprints is smaller than $|\Sigma|=T^{3S}$. Let $(L,K,F)$ be a footprint. First, observe that the number of distinct $L$ is at most $4^{2S}$. Indeed, by definition of the space used by $\calM$, the list $L$ has $S$ entries and the total number of bits needed to encode the values of all registers is at most $S$. Hence, one could see $L$ as a word of length at most $2S$ over alphabet $\{0,1,\star,\#\}$, where $\#$ is used to separate the entries of $L$. There exists at most $4^{2S}=16^S$ such words, hence at most $16^S$ distinct $L$.

Now, fix $L$. We claim that there are at most $T^S$ possible values for $K$. It directly comes from the fact that the entry of $K$ that are different from $\bot$ are determined by $L$. There are at most $S$ such entries, all of which containing a value less that $T$. Hence, there are at most $T^S$ possible values for $K$.

Finally, fix $L$ and $K$. We claim that there now exists at most $T^S$ possible values for $F$. Indeed, let $i$ be such that $F[i] \neq []$. By definition, it means that $R_i$ contains a pointer toward a simulation of $\ladv$ for some $\sigma$ for which $K[i] \neq \bot$ steps have been executed. First assume $K[i] \geq T_0$. All solutions that $\ladv$ will output from this point are spaced by $3DS$ steps by definition. Hence, moving $\ladv$ by $3DS$ steps will generate at most one output of the form $\sigma(j)\odot[T_0+j]_2$ for some $j$. The value of $j$ is completely determined by $K[i]$, hence $\sigma(j)$ is the only parameter here and can take at most $T$ distinct values by definition. Since there are at most $S$ such entries in $F$, there are at most $T^S$ distinct $F$.

In other words, there are at most $4^{2S} \cdot T^S \cdot T^S = (4T)^{2S}$ distinct footprints. Since there exists $T^{3S} > (4T)^{2S}$ distinct lists in $\Sigma$, there necessarily exist distinct indistinguishable lists in $\Sigma$, which is, from what precedes, a contradiction with the assumption that $N < 9d(p,b)s(p,b)^2$.
\end{proof}

\begin{proof}[Proof of \Cref{thm:lower-bound-ordered}]
  By \Cref{lem:lower-bound-order}, for every $N \in \N$, $9d(2,2\log N+1) \times s(2,2\log N+1)^2 > N$. If both $d$ and $s$ are polynomial in $b$ then it means that there exists a polynomial $P(X) > 9d(2,2X+1) \times s(2,2X+1)^2$ such that $P(\log N) > N$ for every $N$ which is a contradiction. 
\end{proof}

On particular consequence of \Cref{thm:lower-bound-ordered} is that one cannot use regularization schemes to prove that classes $\DelayP^\poly$ and $\AmDelayP^\poly$ are the same for all possible orders. Indeed, the problems considered in these classes being in $\EnumP$, the size of solution is polynomial in the size of the input. Hence \Cref{thm:lower-bound-ordered} states that either the delay or the space used by a regularization scheme on such problems will have a behavior exponential in the input size. 

%%% Local Variables:
%%% mode: latex
%%% TeX-master: "theoretics"
%%% End:

\section{Oracles to RAM}
\label{sec:oracles}

In this section, we explain how to implement the blackbox operations in our regularization schemes in constant time when the enumeration process $\abbox$ is given by a RAM machine $M$ working on some input $x$. We denote by $n$ the size of $x$. 
% We let $p(n)$ be the delay of $I$, the simulated RAM and $s(n)$ be its \emph{known} space use. We assume here that both $p(n)$ and $s(n)$ are polynomials.
The complexity of any operation in the RAM model, say $a + b$ is $(\log(a) + \log(b)) / \log(n)$. If $a$ and $b$ are bounded by some polynomial in $n$, then $(\log(a) + \log(b)) / \log(n) < C$ for some constant $C$. All integers used in this section are bounded by a polynomial in $n$ and can thus be manipulated in constant time and stored using constant space.

We assume an infinite supply of \emph{zero-initialized memory}, that is all registers of the machines we use are first initialized to zero. It is not a restrictive assumption, since we can relax it, by using a lazy initialization method (see~\cite{mehlhorn2013data} 2, Section III.8.1) for all registers, for only a constant time and space overhead for all memory accesses. 

\subsection{Pointers and Memory}

To implement extensible data structures, we need to use pointers. 
A pointer is an integer, stored in some register, which denotes the index of the register from which is stored an element. In this article, the value of a pointer is always bounded by a polynomial in $n$, thus it requires constant memory to be stored.
Using pointers, it is easy to implement linked lists, each element contains a pointer to its value and a pointer to the next element of the list. Following a pointer in a list can be done in constant time. Adding an element at the end of a list can be done in constant time if we maintain a pointer to the last element.
We also use arrays, which are a set of consecutive registers of known size.

In our algorithms, we may need memory to extend a data structure or to create a new one, but we never need to free the memory. 
Such a memory allocator is trivial to implement: we maintain a register containing the value $F$, such that no register of index larger than $F$ is used. When we need $k$ consecutive free registers to extend a data structure, we use the registers from $F$ to $F+k-1$ and we update $F$ to $F+k$.

\subsection{Counters}
\label{sec:counters}

All algorithms presented in this paper rely, sometimes implicitly, on our ability to efficiently maintain counters, in particular to keep track of the number of steps of a RAM that have been simulated so far. Implementing them naively by simply incrementing a register would result in efficiency loss since these registers may end up containing values as large as $2^{\poly(n)}$ and we could not assume that this register can be incremented, compared or multiplied in constant time in the uniform cost model that we use in this paper.

To circumvent this difficulty, we introduce in this section a data structure that allows us to work in constant time with counters representing large values. Of course, we will not be able to perform any arithmetic operations on these counters. However, we show that our counter data structure enjoys the following operations in constant time: $\inc(c)$ increases the counter by $1$ and $\mbit(c)$ returns the index of the most significant bit of the value encoded by $c$. In other words, if $k = \mbit(c)$ then we know that $\inc(c)$ has been executed at least $2^k$ times and at most $2^{k+1}$ times since the initialization of the counter.

The data structure is based on Gray code encoding of numbers. A Gray code is an encoding enjoying two important properties: the Hamming distance of two consecutive elements in the Gray enumeration order is one and one can produce the next element in the order in constant time. The method we present in this section is inspired by Algorithm G presented in~\cite{knuth2011combinatorial} which itself is inspired by~\cite{bitner1976efficient} for the complexity. The only difference with Algorithm G is that we maintain a stack containing the positions of the $1$-bits of the code in increasing order so that we can retrieve the next bit to switch in constant time which is not obvious in Algorithm G. Our approach is closer to the one presented in Algorithm L of~\cite{knuth2011combinatorial} but for technical reasons, we could not use it straightforwardly.

We assume in the following that we have a data structure for a stack supporting initialization, $\push$ and $\pop$ operations in constant time and using $O(s)$ registers in memory where $s$ is the size of the stack (it can be implemented by a linked list). 

\subparagraph{Counters with a known upper bound on the maximal value.} We start by presenting the data structure when an upper bound on the number of bits needed to encode the maximal value to be stored in the counter is known. For now on, we assume that the counter will be incremented at most $2^{k}-1$ times, that is, we can encode the maximal value of the counter using $k$ bits.

To initialize the data structure, we simply allocate $k$ consecutive registers $R_0, \dots, R_{k-1}$ initialized to $0$, which can be done in constant time since the memory is assumed to be initialized to $0$, and we initialize an empty stack $S$. Moreover, we have two other registers $A$ and $M$ initialized to $0$.

We will implement $\mbit$ and $\inc$ to ensure the following invariants: the bits of the Gray Code encoding the value of the counter are stored in $R_0, \dots, R_{k-1}$. $A$ contains the parity of the number of $1$ in $R_0, \dots, R_{k-1}$. $M$ contains an integer smaller than $k$ that is the position of the most significant bit in the Gray Code (the biggest $j \leq k-1$ such that $R_j$ contains $1$). Finally, $S$ contains all positions $j$ such that $R_j$ is set to $1$ in decreasing order (that is if $j < j'$ are both in $S$, $j$ will be poped before $j'$).

To implement $\mbit$, we simply return the value of $M$. It is well-known and can be easily shown that the most significant bit of the Gray Code is the same as the most significant bit of the value it represents in binary so if the invariant is maintained, $M$ indeed contains a value $j$ such that the number of time $\inc(c)$ has been executed is between $2^j$ and $2^{j+1}-1$.

To implement $\inc$, we simply follow Algorithm G from~\cite{knuth2011combinatorial}. If $A$ is $0$ then we swap the value of $R_0$. Otherwise, we swap the value of $R_{j+1}$ where $j$ is the smallest position such that $R_j=1$ (if $j$ is $k-1$ then we have reached the maximal value of the code which we have assumed to be impossible, see below to handle unbounded counters). One can find $j$ in constant time by just popping the first value in $S$, which works if the invariant is maintained. Now, one has to update the auxiliary memory: $A$ is replaced by $1-A$ so that it still represents the parity of the number of $1$ in the Gray Code. To update $S$, we proceed as follows: if $A$ is $0$ then either $R_0$ has gone from $0$ to $1$, in which case we have to push $0$ in $S$ or $R_0$ has gone from $1$ to $0$, in which case we have to pop one value in $S$, which will be $0$ since $S$ respects the invariant. It can be readily proven that this transformation preserves the invariant on $S$. Now, if $A$ is $1$, then either the value of $R_{j+1}$ has gone from $0$ to $1$ which means that we have to push $j+1$ and $j$ on the stack ($j$ is still the first bit containing $1$ so it has to be pushed back on the top of the stack and $j+1$ is the next bit set to $1$ so it has to be just after $j$ in $S$). Or the value of $R_{j+1}$ has gone from $1$ to $0$. In this case, it means that after having popped $j$ from $S$, $j+1$ sits at the top of $S$. Since $R_{j+1}$ is not $0$, we have to pop $j+1$ from $S$ and push back $j$. Again, it is easy to see that these transformations preserve the invariant on $S$. Moreover, we never do more than $2$ operations on the stack so this can be done in constant time. Finally, if $R_{j+1}$ becomes $1$ and $j+1>M$, we set $M$ to $j+1$. 

Observe that we are using $2k+2$ registers for this data structure since the stack will never hold more than $k$ values. 

\subparagraph{Unbounded counters.} To handle unbounded counters, we start by initializing a bounded counter $c_0$ with $k$ bits ($k$ can be chosen arbitrarily, $k=1$ works). When $c_0$ reaches its maximal value, we just initialize a new counter $c_1$ with $k+1$ bits and modify it so it contains the Gray Code of $c_0$ (with one extra bit) and copy its stack $S$ and the values of $A$ and $M$.

This can be done in constant time thanks to the following property of Gray code: the Gray code encoding of $2^k-1$ contains exactly one bit set to $1$ at position ${k-1}$. Thus, to copy the value of $c_0$, we only have to swap one bit in $c_1$ (which has been initialized to $0$ in constant time). Moreover, the stack of $c_0$ containing only positions of bit set to $1$, it contains at this point only the value $k-1$ that we can push into the stack of $c_1$. Copying registers $A$ and $M$ is obviously in constant time. 

An example for a Gray code using three bits is described below:
\begin{itemize}
 
\item $c_0$:
  \begin{itemize}
  \item Bits: 000, Stack: []
  \item Bits: 001, Stack: [0]
  \item Bits: 011, Stack: [0,1]
  \item Bits: 010, Stack: [1]
  \item Bits: 110, Stack: [1,2]
  \item Bits: 111, Stack: [0,1,2]
  \item Bits: 101, Stack: [0,2]
  \item Bits: 100, Stack: [2]
  \end{itemize}

\item $c_1$ :
  \begin{itemize}  
  \item Bits: 0000, Stack: [] (initialization)
  \item Bits: 0100, Stack: [2] (copy the value of $c_0$)
  \item $\dots$
  \end{itemize}
\end{itemize}

To summarize, we have proven the following:

\begin{theorem}
  \label{th:counter} There is a data structure $\Counter$ that can be initialized in constant time and for which operations $\inc$ and $\mbit$ can be implemented in constant time with the following semantic: $\mbit(c)$ returns an integer $j$ such $v$ is between $2^{j}$ and $2^{j+1}-1$ where $v$ is the number of time $\inc(c)$ has been executed since the initialization of $c$. Moreover, the data structure uses $O(\log(v)^2)$ register.
\end{theorem}

One could make the data structure more efficient in memory by lazily freeing the memory used by the previous counters so that it is $O(\log(v))$. However, such an optimization is not necessary for our purpose.

\subsection{Instructions $\ramload$, $\ramstep$ and $\steps$}%for Known Parameters
\label{sec:countersknown}

In this section, we explain formally how one can use as an oracle an enumeration process $\abbox$ given by a RAM $I$ and an input $x$. We assume that we know an upper bound $s(n)$ on the memory used by $I$ on inputs of size $n$ and that $s(n)$ is polynomial in $n$. We also assume that  the amortized delay $p(n)$ of $I$ on inputs of size $n$ is polynomially bounded though we do not assume it to be known in advance. We also assume that the solutions output by $I$ are of size polynomial in $n$. In other words, if $N(n)$ is the maximal number of solutions output by $I$ on inputs of size $n$, then $\log(N(n))$ is bounded by a polynomial in $n$. This assumption always hold from machine solving $\EnumP$ problems, the setting we consider in this paper. Recall that the total time of $I$ on input $x$ is bounded by $tt(n) := N(n) \cdot p(n)$.

We explain how one can implement the blackbox instructions $\bbload$, $\bbstep$ and $\bbsteps$ with complexity $O(1)$. To simulate $\abbox$, defined as the enumeration process of $I$ on input $x$, we will maintain in memory the \textbf{configuration} of $\abbox$ after $k$ steps. The configuration of $I$ on input $x$ at step $k$ is defined as the content of the registers $R_1, \dots, R_{s(n)}$ after $k$ steps, the index of the next instruction to be executed by $I$ and a counter $c$ representing $k$, encoded as in \Cref{sec:counters}. To store a configuration, we hence need $O(s(n)+\log(tt(n))^2)$ registers, the instruction index being always bounded by a constant, it can be stored in one register. 

Instruction $\bbload$ hence reserves $s(n)+2$ consecutive registers in memory: enough to store the memory of $\abbox$ and two extra registers, one holding the instruction index (whose value will be bounded by a constant depending on $I$) and one holding a pointer to a dynamic Gray code base counter as in \Cref{sec:counters} representing the number of steps executed so far. The instruction returns $M$, the index of the first register storing the configuration of $I$ and initializes the step counter to $0$ and the instruction index to $0$. Instruction $\bbstep$ executes the command given by the instruction index by offseting register numbers by $M$. It then changes the instruction index following the program flow and increments the step counter $c$.  Instruction $\bbisoutput$ returns true only if the instruction index points on an output instruction $\Output(i,j)$ of the RAM machine. If so, instruction $\bboutput$ returns $(i+M, j+M)$, indicating that the current output of $\abbox$ is stored in registers $R_{i+M}, \dots, R_{j+M}$. Finally, instruction $\bbsteps$ returns the value of the maintained step counter. 

All of this can be done in constant time, the only non trivial part being incrementing the counter which can grows too much. However, using Gray Code based counters from \Cref{sec:counters}, one can circumvent this difficulty. The total number of registers used by this representation will never exceed $O(s(n)+\log(tt(n))^2)$ and hence is polynomial in $n$ following the assumption made earlier in this section.

\subsection{Implementing regularization schemes}
\label{sec:implementing}

While we just show that every black box operation can be simulated in $O(1)$ when the black box corresponds to some RAM machine $I$ on input $x$ of size $n$, it is however not sufficient to have an efficient implementation of the different schemes we have presented so far. The reason for this is that in every algorithm, we need to either compare $\bbsteps$ with some value or to do some arithmetic computation on it, e.g., Line~\ref{line:test} in \Cref{alg:enumit} tests whether $\bbsteps[M[j]]$ is inside some set $Z_j$. Since $\bbsteps$ may be exponential in $|x|$ at some point, we cannot expect this test to be executed in constant time. We explain in this section how one can circumvent this kind of difficulties in the implementation.

In this section, we suppose that $\bbam = p(n)$  is polynomial in $n$ and that the space $s(n)$ used by $\abbox$ is also polynomial in $n$. 

\paragraph*{Zone detection for \Cref{alg:enumit}, \Cref{alg:enumit_improved} and \Cref{alg:geometric-adaptative} }

 In these three algorithms, we must test whether $M[j]$ satisfies $\steps(M[j]) \in $ $[{2^{j-1} p(n)}+1, {2^j p(n)}]$, e.g., in Line~\ref{line:test} of Algorithm~\ref{alg:enumit} ($p(n)$ is replaced by $F$, a precomputed lower bound in \Cref{alg:geometric-adaptative}). To check this inequality in constant time, we simply initialize a counter $c_j$ as in Section~\ref{sec:counters}. Instead of incrementing it each time $\move(M[j])$ is called, we increment it every $p(n)$ calls to $\move$. This can easily be done by keeping another register $R$ which is incremented each time $\move(M[j])$ is called and whenever it reaches value $p(n)$, it is reset to $0$ and $c_j$ is incremented. 
 To decide whether $M[j]$ enters its zone, it is sufficient to test whether $\mbit(c_j) = j-1$. The first time it happens, then exactly $2^{j-1}p(n)$ steps of $M[j]$ have been executed, so it will enters its zone in the next move, so we can remember it to start the enumeration. When $\mbit(c_j)$ becomes $j$, it means that $2^jp(n)$ steps of $M[j]$ have been executed, that is, $M[j]$ leaves its zone. Thus, we can perform the comparison in constant time.

 \paragraph*{Approximate budget in \Cref{alg:enumqueue_improved}, \Cref{alg:geometric-adaptative} and \Cref{alg:enumitinc}} 
 %In these three algorithms, one has to compare a variable $b$ with some function of $d$, which is a local bound on the delay. 

In \Cref{alg:enumqueue_improved}, \Cref{alg:geometric-adaptative}, the difficulty is to initialize the variable $B$ to $d\log(d)^2$. The variable $d$ is a local bound on the delay and is updated after every output by $\max (d, \frac{\steps(M)}{1+S_1})$.
However, even if $\steps(M) \over 1+S_1$ is bounded by $p(n)$, one cannot compute it in constant time since both $\steps(M)$ and $S_1$ can reach $O(2^{\poly(n)})$ during the execution of the algorithm. 

To circumvent this complexity, we compute an overapproximation of $B$, which is close enough to $B$ and can then be decremented in constant time since its value is bounded by $Cp(n)$ for some constant $C$. We encode $\steps(M)$ and $S_1$ with two counters as presented in Section~\ref{sec:counters}. We compute in constant time
their most significant bits $k_M = \mbit(\steps(M))$ and $k_S = \mbit(S_1)$. 

We have that $$2^{k_M} \leq \steps(M) < 2^{k_M +1}$$ and $$2^{k_S} \leq S_1 < 2^{k_S +1}.$$

 Hence, we obtain the following inequalities: $$2^{k_M-k_S -1} \leq  \frac{\steps(M)}{1+S_1} \leq 2^{k_M-k_S +1}.$$ 

 We maintain the value $R$ instead of $d$, defined as $R = \max(R,2^{k_M-k_S +1})$ and we remark that $R \geq d$ and $R \leq 4d$ because of the previous inequalities.
Moreover, $R$ can be computed in constant time from $k_M$ and $k_S$, since its value is bounded by $4d \leq 4p(n)$. The value of $\log(R)$ is $k_M-k_S +1$ and may also be obtained in constant time. Instead of setting $B$ to $d\log(d)^2$ in the algorithm, we set it to $R\log(R)^2$. Since $R \geq d$, the proof of correction of both algorithms is still correct. Morevoer, since  $R\log(R)^2$ is larger than $d\log(d)^2$ by less than a factor of $16 = 4\cdot \log(4)^2$, the delay of the two algorithm is at most multiplied by this constant factor.  

In \Cref{alg:enumitinc}, we have a similar problem: while $d$ is given at the beginning of the algorithm, we must compare the variable $B$ to $2d\cdot S^a(a+1)$, where $S$ can reach $O(2^{\poly(n)})$ during the execution of the algorithm. We use the same method of approximation of $S$ by its most significant bit. Hence, we achieve a constant time approximate comparison, which does not change the correction of the algorithm and increases its delay by a constant factor only.

\paragraph*{Lazy implementation of $\ramcopy$}

\label{sec:copyorcales}

Algorithm~\ref{alg:enumit_improved} does not use the $\bbsol$ upper bound, but instead relies on the $\bbcopy$ instruction. This instruction takes as parameter a data structure $M$ storing the configuration of a RAM and returns a new data structure $M'$ of the same machine starting in the same configuration (an exact copy of the memory). A straightforward way of implementing $\bbcopy$ would be to copy every register used by the data structure $M$ in a fresh part of the memory. However, this approach may be too expensive, since we need to copy the whole memory used by $M$. Since we are guaranteed to have one output solution between each $\bbcopy$ instruction, the delay of Algorithm~\ref{alg:enumit} then becomes $O(\log(\bbsol)(\bbam + s(n)))$.

We explain how one can lazily implement this functionality so that the memory of $M$ is copied only when needed. To avoid having a preprocessing in $O(s(n))$ we track the indices of the used registers in the beginning of the enumeration to do only a sparse copy. This method ensures that $\bbcopy$ runs in $O(1)$ with only a constant overhead in the other instructions and space. 

As a consequence of all the implementations of the blackbox instructions described in this section, we give a version of \Cref{th:geometric_improved}, where the enumeration process is induced by a RAM. Similar Theorems can be proved for the complexity and space of all enumeration schemes of this article applied to RAMs.

\begin{theorem}\label{th:space_not_cardinal}
Let $I$ be a RAM. Let $T(n)$ be the total time of $I$, $p(n)$ and $s(n)$ are two given polynomials which respectively upperbound the amortized delay and the space of $I$. We can construct a RAM $I'$ which enumerates $I(x)$ on input $x \in \Sigma^n$ with:
\begin{itemize}
\item delay $O(p(n)\log(\#I(x)))$,
\item space $O(s(n)\log(\#I(x)))$,
\item total time $O(T(n))$. 
\end{itemize}
\end{theorem}

\begin{proof}
We use the regularization scheme of \Cref{alg:enumit_improved}, whose complexity is stated in \Cref{th:geometric_improved}. To be able to state the time and space used when using an enumeration process derived from a RAM machine, we must analyze the time complexity and space complexity of the implementations of each instruction.

We have already explained in the section how to implement all blackbox instructions but $\bbcopy$ in time $O(1)$. We first need to evaluate the space used by the $\lceil \log(\#I(x)) \rceil$  simulations of $I$. RAM $I$ uses a space $s(n)$ and we have explained how a simulation of $I$ uses a space $O(s(n))$. Moreover, for each simulation we maintain a Gray code counter as explained in \Cref{sec:counters}. In \Cref{alg:enumit_improved}, we know the maximal value of a counter when it is created by a $\bbcopy$ instruction, hence we need to allocate the memory only once. Since the maximal value of a counter is always bounded by $\lceil \#I(x) \rceil$, the memory used is $O(\log(\#I(x)))$. The space used has to be larger than the logarithm of the total time of $I$ on $x$, otherwise $I$ would be in the same configuration at two points in time and loop.
Since $\#I(x)$ is less than the total time and $s(n)$ larger than the space used, $s(n) \geq \log(\#I(x))$. Hence, each simulation uses a space $O(s(n))$ and the space used by the algorithm to store counters is $O(s(n)\log(\#I(x)))$.

Let us now explain how the data structure $M$ is lazily copied when the instruction $\bbcopy[M]$ is executed. The copied data structure contains a register for the index of the current instruction, a counter of the number of steps and an array to represent the registers of $I$ the simulated machine.

We build the new data structure $M'$ given by $\bbcopy[M]$ as follows. 
The counter in $M'$ is stored as in Theorem~\ref{th:counter}. It is initialized so that it represents the value $2^{j-1}p(n)$ and it counts up to $2^{j+1}p(n)$. This value is represented by a classical counter of value $0$ stored in a single register and the Gray code counter contains the  $2^{j-1}$th integer in Gray code order for integers of size $j+1$.
This number is equal to $2^{j-1} + 2^{j-2}$, which has only two one bits (the second and the third), hence it can be set up in constant time. The auxiliary structure is
the list of ones of the integer, which is here of size two and can thus be set up in constant time.

We explain how to lazily create an exact copy of the array representing the state of the registers at the time of the copy. Let us denote this array by $A$, it is of size $s(n)$. We create in constant time $A'$ and $U$ of size $s(n)$ both initialized to zero. The value of $U[r]$ is $0$ if $A'[r]$ has not yet been copied from $A[r]$ and $1$ otherwise. Each time $\ramstep(M)$ is executed and modifies the value $A[r]$, if $U[r] = 0$, it first set $A'[r] = A[r]$ and $U[r] = 1$. Each time $\ramstep(M')$ is executed and reads the value $A'[r]$, if $U[r] = 0$, it first set $A'[r] = A[r]$ and $U[r] = 1$.
This guarantees that the value of $A'$ is always the same as if we had completely copied it from $A$ when the instruction $\bbcopy$ is executed. 
The additional checks and updates of $U$ add a constant time overhead to $\ramstep$. 
Moreover, we maintain a simple counter $c$, and each time a $\ramstep(M')$ operation is executed, if $U[c] = 0$, we set $A'[c] = A[c]$ and $U[c] = 1$. When $c=s(n)$, the copy is finished and we can use $A$ and $A'$ as before, without checking $U$.

The described implementation of the $\bbcopy$ operation is in constant time. The $\ramstep$ instruction, modified as described, has a constant overhead for each lazy copy mechanism in action. To evaluate the complexity of \Cref{alg:enumit_improved}, we must evaluate the number of active copies. We prove that \Cref{alg:enumit_improved} has only a single active copy mechanism at any point in time, with an improvement to the copy mecanism.

 We maintain in the algorithm an additional array $S$ of size $s(n)$ which contains the indices of registers written by $I$. The indices are written in the order they appear in the computation and we may have the same index stored twice. We also maintain $r$ the number of stored indices. When $r$ is larger than $s(n)$, we do not update this array anymore.  Note that it is possible to maintain this array in constant time when doing a $\ramstep(M)$ operation, since the values inside the array are bounded by $s(n)$.

 When the $\bbcopy$ operation is executed, we store the current value of $r$ in $r_{last}$. When $r_{last} \geq s(n)$, we use the implementation of $\ramstep(M)$ and $\bbcopy$ previously defined.

 When $r_{last} \leq s(n)$, instead of doing $A'[c] = A[c]$, we do $A'[S[c]] = A[S[c]]$, that is we copy only the registers which have been written by $I$ up to this point. Hence, we need to do this copy until $c = r_{last}$ instead of $c = s(n)$.

Now, consider the $\bbcopy$ operation which creates the data structure $M[j]$ at the beginning of the zone $Z_{j-1}$ in \Cref{alg:enumit_improved}. The number of simulation steps of $I$ before $M[j-1]$ arrives at $Z_{j-1}$ is $2^{j-1}d$. As a consequence, when $M[j]$ is created we have $r_{last} \leq 2^{j-1}d$. Since the size of $Z_{j-1}$ is $2^{j-1}d$, when the $M[j]$ goes through the zone $Z_{j-1}$, we have that $c = r_{last}$ at some point and the copy is finished. 
Therefore, there is a single lazy copy at each point of \Cref{alg:enumit_improved}.

We have proved that the $\bbcopy$ instruction can be executed in time $O(1)$ and space $O(s(n))$, with a constant overhead to the other instructions which proves the time and space complexity bounds of this theorem.
\end{proof}

%%% Local Variables:
%%% mode: latex
%%% TeX-master: "theoretics"
%%% End:

\section{Conclusion}
\label{sec:conclusion}
  
In this paper, we exhaustively studied the problem of transforming algorithms with small amortized delay into algorithms having a small worst case delay. While a naive buffer based regularization method has been folklore in the literature, we propose several ways of improving it, by reducing the memory footprint or by automatically adjusting to unknown amortized delay. While, we have ignored the fact that the blackbox may have a preprocessing greater than its amortized delay, we observe here that all our regularization schemes can easily be modified to take this into account because every $\bbload$ instruction happens at the beginning of the algorithm. Hence, one could execute several $\bbstep$ instructions until the first solution is found just after $\bbload$, before starting the enumeration algorithm. This would result in a regularization scheme with preprocessing. 

In this paper, we complemented our algorithms with lower bounds in a model of computation where the underlying enumeration algorithm is accessed in a blackbox fashion. We have shown that no regularization scheme can achieve a worst case delay that is linear in the real amortized delay.  Moreover, we show that if one wants to preserve the order during the enumeration process, then either the space used by the algorithm or the worst case delay has to be exponential in the size of the solutions. In most enumeration problems, it means that either the space or the delay have to be exponential in the size of the input.

That said, all our techniques and lower bounds are dealing with enumeration algorithms in a blackbox manner. An interesting research avenue could be to explore the limits of regularization methods that are allowed to access enumeration algorithms with their code. For example, with our method, we cannot rule out that every problem in $\AmDelayP^\poly$ can be turned into a $\DelayP^\poly$ problem while preserving the enumeration order. Our result only shows that if it were possible, then the proof should somehow rely on internal knowledge of the enumeration algorithm. We leave this question open for future work.

% Say something about the implementation of the blackbox operations
% when s(n) is not known and poly, when p(n) is not poly ?
%%% Local Variables:
%%% mode: latex
%%% TeX-master: "theoretics"
%%% End:

\bibliography{biblio}

% \section*{Appendix}
% \label{sec:appendix}

% \input{appendix}

\end{document}